\documentclass[twocolumn]{IEEEtran}

\usepackage{amsmath,epsfig,amssymb,verbatim,amsopn,cite}
\usepackage{amsthm}
\usepackage{balance}
\usepackage{multirow}
\usepackage[usenames,dvipsnames]{color}
\usepackage[all]{xy}  
\usepackage{url}
\usepackage{amsfonts}
\usepackage{amssymb}
\usepackage{epsfig}
\usepackage{epstopdf}
\usepackage{bm}
\usepackage{balance}
\usepackage{graphicx}
\usepackage{subcaption}
\usepackage{footnote}
\usepackage{cancel}
\usepackage{algorithm,algorithmic}
\usepackage{longtable}
\usepackage{array}

\usepackage{multicol} 
\usepackage{makecell} 
\usepackage[short]{optidef} 

\usepackage{tikz}
\newcommand{\tikzxmark}{%
\tikz[scale=0.15] {
    \draw[line width=0.6,line cap=round] (0,0) to [bend left=5] (1,1);
    \draw[line width=0.6,line cap=round] (0.1,0.85) to [bend right=3] (0.8,0.05);
}}

\usepackage[flushleft]{threeparttable}

\newtheorem{Remark}{Remark}

  {\proof}{\proofend}
\newtheorem{proposition}{Proposition}

\usepackage[margin=15.1mm,top=17.8mm,bottom=25.5mm]{geometry}

\newcommand{\qa}{\textbf{a}}

\newcommand{\qe}{\textbf{e}}

\newcommand{\qg}{{ \textbf{g} }}

\newcommand{\qn}{{\bf n}}

\newcommand{\qs}{{\bf s}}

\newcommand{\qw}{{\bf w}}
\newcommand{\qx}{{\bf x}}
\newcommand{\qy}{{ \textbf{y} }}
\newcommand{\qz}{{ \textbf{z} }}

\newcommand{\qA}{{\bf A}}
\newcommand{\qB}{{\bf B}}
\newcommand{\qC}{{\bf C}}
\newcommand{\qD}{{\bf D}}

\newcommand{\qF}{{\bf F}}
\newcommand{\qG}{{ \textbf{G} }}
\newcommand{\qH}{{ \textbf{H} }}
\newcommand{\qI}{{ \textbf{I} }}

\newcommand{\qN}{{\bf N}}

\newcommand{\qS}{{\bf S}}

\newcommand{\qW}{{\bf W}}

\newcommand{\qY}{{\bf Y}}
\newcommand{\qZ}{{ \textbf{Z} }}

\newcommand{\ettall}{\emph{et al.}}

\newcommand{\UE}{\mathtt{I}}
\newcommand{\sn}{\mathtt{E}}

\DeclareMathOperator{\ETAI}{\boldsymbol{\eta}^{\mathtt{I}}}
\DeclareMathOperator{\ETAE}{\boldsymbol{\eta}^{\mathtt{E}}}

\DeclareMathOperator{\K}{\mathcal{K}}

\DeclareMathOperator{\M}{\mathcal{M}}

\DeclareMathOperator{\vecOp}{\mathrm{vec}}

\DeclareMathOperator{\C}{\mathbb{C}}
\DeclareMathOperator{\R}{\mathbb{R}}

\DeclareMathOperator{\CN}{\mathcal{CN}}

\newcommand{\PZF}{\mathsf{ZF}}
\newcommand{\PMRT}{\mathsf{PMRT}}

\newcommand{\wimk}{\qw_{\mathtt{I},mk_{i}}}

\newcommand{\wemj}{\qw_{\mathtt{E},m k_{e}}}
\newcommand{\wemjp}{\qw_{\mathtt{E},m k'_{e}}}

\newcommand{\wimkp}{\qw_{\mathtt{I},mk_{i}'}}

\newcommand{\Ghms}{\hat{\qG}_m^{\sn}}

\newcommand{\Ghmu}{\hat{\qG}_m^{\UE}}

\newcommand{\hatWi}{\hat{\qW}^{\mathtt{I}, \circ}_{m}}
\newcommand{\hatWe}{\hat{\qW}^{\mathtt{E}, \circ}_{m}}

\newcommand{\Snn}{\sigma_n^2}
\newcommand{\Ex}{\mathbb{E}}

\newcommand{\yej}{y_{\mathtt{E}, k_e}}

\newcommand{\yik}{y_{\mathtt{I},k_i}}

\newcommand{\gmkiu}{\qg_{mk_{i}}^{\UE}}

\newcommand{\gmKiuH}{\qg_{mK}^{\UE H}}
\newcommand{\gmjue}{\qg_{mk_{e}}^{\sn}}

\newcommand{\hgmjue}{\hat{\qg}_{m k_e}^{\sn}}

\newcommand{\tildehgmki}{\tilde{\hat{\qg}}_{m k_i}^{\UE}}

\newcommand{\hgmjueH}{(\hat{\qg}_{m k_e}^{\sn})^\dag}

\newcommand{\FmRIS}{\qF_{m}}

\newcommand{\gejue}{\qg_{m k_{e}}^{\sn}}
\newcommand{\gejueH}{\left(\qg_{m k_{e}}^{\sn} \right)^{\!\dag}}

\newcommand{\trace}{\mathrm{tr}}
\newcommand{\diag}{\mathrm{diag}}

\newcommand{\hgmkue}{\hat{\qg}_{mk}^{\UE}}

\newcommand{\gtilmjeu}{\tilde{\qg}_{m k_e }^{\sn}}
\newcommand{\gtilmjeuH}{(\tilde{\qg}_{m k_e}^{\sn})^{\dag}}

\newcommand{\gameumj}{\gamma_{m k_e}^{\sn}}
\newcommand{\gameumjp}{\gamma_{m k'_e}^{\sn}}

\newcommand{\snrul}{\rho_{u}}
\newcommand{\snrdl}{\rho_{d}}

\newcommand{\zmk}{\qz^{\mathtt{i}}_{mk}}
\newcommand{\zmklos}{\bar{\qz}_{mk}^{\mathtt{i}}}
\newcommand{\zmkilos}{{\qz}_{m k_i}^{\mathtt{I}} }
\newcommand{\zmkiplos}{{\qz}_{m k'_i}^{\mathtt{I}} }
\newcommand{\zmkelos}{\bar{\qz}_{m k_e}^{\mathtt{E}} }
\newcommand{\zmkeplos}{\bar{\qz}_{m k'_e}^{\mathtt{E}} }

\newcommand{\zmknlos}{\tilde{\qz}_{mk}^{\mathtt{i}}}
\newcommand{\zmkpnlos}{\tilde{\qz}_{mk'}^{\mathtt{i}}}

\newcommand{\gmk}{\qg_{mk}^{\mathtt{i}}}

\newcommand{\gmkbar}{\bar{\qg}_{mk}^{\mathtt{i}}}
\newcommand{\gmkebar}{\bar{\qg}_{m k_e}^{\mathtt{E}}}
\newcommand{\gmkepbar}{\bar{\qg}_{m k'_e}^{\mathtt{E}}}

\newcommand{\gmkpbar}{\bar{\qg}_{mk'}^{\mathtt{i}}}
\newcommand{\gmkibar}{\bar{\qg}_{m k_i}^{\mathtt{I}}}

\newcommand{\hgmk}{\hat{\qg}_{mk}^{\mathtt{i}}}
\newcommand{\hgmki}{\hat{\qg}_{m k_i}^{\mathtt{I}}}

\newcommand{\tilgmk}{\tilde{\qg}_{mk}^{\mathtt{i}}}
\newcommand{\tilgmki}{\tilde{\qg}_{m k_i}^{\mathtt{I}}}

\newcommand{\betamk}{\beta_{mk}^{\mathtt{i}}}
\newcommand{\barbetamk}{\bar{\beta}^{\mathtt{i}}_{mk}}
\newcommand{\barbetamkp}{\bar{\beta}^{\mathtt{i}}_{mk'}}
\newcommand{\barbetamki}{\bar{\beta}^{\mathtt{I}}_{m k_i}}

\newcommand{\barbetamke}{\bar{\beta}^{\mathtt{E}}_{m k_e}}
\newcommand{\barbetamkep}{\bar{\beta}^{\mathtt{E}}_{m k'_e}}

\newcommand{\ymk}{\qy_{mk}^{\mathtt{i}}}

\newcommand{\DSki}{\mathrm{DS}_{k_i}}
\newcommand{\BUki}{\mathrm{BU}_{k_i}}
\newcommand{\IUIki}{\mathrm{IUI}_{k_{i} {k_{i}}'}}
\newcommand{\EUIki}{\mathrm{EUI}_{k_{i} k_{e}}}

\newcommand{\SEth}{\mathcal{S}_{k_i}}

\newcommand{\Covhatgmk}{\boldsymbol{\Sigma}_{\hgmk}}
\newcommand{\Covhatgmki}{\boldsymbol{\Sigma}_{\hat{\qg}_{m k_i}^{\mathtt{I}}}}
\newcommand{\Covhatgmke}{\boldsymbol{\Sigma}_{\hat{\qg}_{m k_e}^{\mathtt{E}}}}

\newcommand{\SINRki}{\mathrm{SINR}_{k_i}}

\newcommand{\Kii}{\mathcal{K}_{\mathtt{I}}}
\newcommand{\Kee}{\mathcal{K}_{\mathtt{E}}}
\newcommand{\varrhomkI}{\varrho_{m k_{i}}^{\mathtt{I}}}
\newcommand{\varrhomkpI}{\varrho_{m k'_{i}}^{\mathtt{I}}}

\newcommand{\varrhomjE}{\varrho_{m k_{e}}^{\mathtt{E}}}
\newcommand{\varrhomjpE}{\varrho_{m k'_{e}}^{\mathtt{E}}}

\newcommand{\etamkI}{\eta_{m k_{i}}^{\mathtt{I}}}
\newcommand{\etamkIn}{\eta_{m k_{i}}^{\mathtt{I}(n)}}

\newcommand{\etamkpI}{\eta_{m k_{i}'}^{\mathtt{I}}}
\newcommand{\etamjE}{\eta_{m k_{e}}^{\mathtt{E}}}

\newcommand{\etamjpE}{\eta_{m k_{e}'}^{\mathtt{E}}}

\newcommand{\etamI}{\eta_{m}^{\mathtt{I}}}
\newcommand{\etamE}{\eta_{m}^{\mathtt{E}}}

\newcommand{\SEk}{\mathrm{SE}_{k_i}}

\newcommand{\xik}{x_{\mathtt{I},k_{i}}}
\newcommand{\xikp}{x_{\mathtt{I},k'_i}}
\newcommand{\xej}{x_{\mathtt{E},k_{e}}}
\newcommand{\xejp}{x_{\mathtt{E},k'_e}}

\newcommand{\alphaPZFmki}{\alpha_{m k_i}^{\PZF}}
\newcommand{\alphaPZFmkip}{\alpha_{m k'_i}^{\PZF}}

\newcommand{\alphaPMRTmke}{\alpha_{m k_e}^{\PMRT}}

\newcommand{\rawlocalactionm}{\boldsymbol{\breve{\Pi}}_{m}}
\newcommand{\localactionm}{\boldsymbol{\Pi}_{m}}

\newcommand{\localstatem}{\boldsymbol{\vartheta}_{m}}

\DeclareMathOperator{\PHI}{\boldsymbol{\Phi}}

\DeclareMathOperator{\VARPHI}{\boldsymbol{\varphi}}

\DeclareMathOperator{\BETA}{\boldsymbol{\beta}}

\begin{document}

\title{Cell-Free Massive MIMO-Assisted SWIPT Using Stacked Intelligent Metasurfaces}

\author{Thien Duc Hua,~\IEEEmembership{Graduate Student Member,~IEEE,}, Mohammadali Mohammadi,~\IEEEmembership{Senior Member,~IEEE,}
\\
Hien Quoc Ngo,~\IEEEmembership{Fellow,~IEEE,}   and  Michail Matthaiou,~\IEEEmembership{Fellow,~IEEE}

\thanks{This work was supported by the U.K. Engineering and Physical Sciences Research Council (EPSRC) grant (EP/X04047X/2) for TITAN Telecoms Hub. The work of T. D. Hua and M. Matthaiou was supported by the European Research Council (ERC) under the European Union’s Horizon 2020 Research and Innovation Programme (grant agreement No. 101001331).}

\thanks{The authors are with the Centre for Wireless Innovation (CWI), Queen's University Belfast, BT3 9DT Belfast, U.K.
(email:\{dhua01, m.mohammadi, hien.ngo, m.matthaiou\}@qub.ac.uk).}
\thanks{ Parts of this paper appeared at the 2025 IEEE ICC~\cite{hua:icc:2025}.
}}

\markboth{.}%
{Shell \MakeLowercase{\textit{et al.}}: A Sample Article Using IEEEtran.cls for IEEE Journals}


\maketitle

\begin{abstract}
This study explores a next-generation multiple access (NGMA) framework for cell-free massive MIMO (CF-mMIMO) systems enhanced by stacked intelligent metasurfaces (SIMs), aiming to improve simultaneous wireless information and power transfer (SWIPT) performance. A fundamental challenge lies in optimally selecting the operating modes of access points (APs) to jointly maximize the received energy and satisfy spectral efficiency (SE) quality-of-service constraints. Practical system impairments, including a non-linear harvested energy model, pilot contamination (PC), channel estimation errors, and reliance on long-term statistical channel state information (CSI), are considered. We derive closed-form expressions for both the achievable SE and the average sum harvested energy (sum-HE). A mixed-integer non-convex optimization problem is formulated to jointly optimize the SIM phase shifts, APs mode selection, and power allocation to maximize average sum-HE under SE and average harvested energy constraints. To solve this problem, we propose a centralized training, decentralized execution (CTDE) framework based on deep reinforcement learning (DRL), which efficiently handles high-dimensional decision spaces. A Markovian environment and a normalized joint reward function are introduced to enhance the training stability across on-policy and off-policy DRL algorithms. Additionally, we provide a two-phase convex-based solution as a theoretical robust performance. Numerical results demonstrate that the proposed DRL-based CTDE framework achieves SWIPT performance comparable to convexification-based solution, while significantly outperforming baselines.

\end{abstract}

\vspace{-1em}
\section{Introduction}
The surging demand for reliable wireless connectivity and efficient energy harvesting (HE) in Internet of Thing devices has made simultaneous wireless information and power transfer (SWIPT) a key research focus. Subsequently, numerous studies have underscored its integration with cell-free massive multiple input multiple output (CF-mMIMO) that could significantly improve both the both HE and spectral efficiency (SE)~\cite{ali:twc:11004488, 10685472}. According to~\cite{2025:Hien:invited, tsp:ris}, while advanced mMIMO is effective at enhancing network capacity, it is often unsustainable due to their excessive energy consumption. Furthermore, the implementation of both massive MIMO and millimeter-wave communications is hindered by the need for expensive hardware components, such as numerous radio frequency (RF) chains and high-resolution digital-to-analog converters. 
These limitations highlight an urgent need for a novel CF-mMIMO SWIPT topology, based on next-generation multiple access techniques (NGMA), that integrates advanced technologies to enable efficient and energy-saving operation. Notably,
a pioneering work by Huang~\ettall~demonstrated the integration of single-layer diagonal reconfigurable intelligent surfaces (RISs) within an access point (AP)'s antenna radome to directly control its signal radiation characteristics~\cite{Huang:TCOM:2023}. While this represents a significant step forward, the inherent structure of RISs imposes limitations relating to the confined degrees of freedom (DoF) for signal manipulation because of their limiting ability to adjust the beam pattern effectively.

To address these challenges, we turn to stacked intelligent metasurfaces (SIMs), an emerging technology that offers a superior architecture. Unlike RISs, the multi-layered structure of a SIM provides additional spatial DoF for signal manipulation compared to RISs. Additionally, SIMs directly manipulate the electromagnetic wave (EM) through a network of reconfigurable layers, leading to the capability of performing complex signal processing tasks~\cite{an:sim:hardware1}. This ability to perform analog computations drastically reduces the need for complex baseband digital precoding, which in turn lowers hardware costs and energy consumption~\cite{2024:an:nearfieldsim}. From a practical standpoints, SIMs emerge as a highly promising and practical candidate for implementing NGMA architectures in CF-mMIMO SWIPT systems.

\begin{table*}
	\centering
	\caption{\label{tabel:Survey} Contrasting our contributions to the SIM-assisted CF-mMIMO literature}
	\vspace{-0.6em}
	\small
\begin{tabular}{|m{2.3cm}|
>{\centering\arraybackslash}m{1.5cm}|
>{\centering\arraybackslash}m{0.63cm}|
>{\centering\arraybackslash}m{0.63cm}|
>{\centering\arraybackslash}m{0.63cm}|
>{\centering\arraybackslash}m{0.63cm}|
>{\centering\arraybackslash}m{0.63cm}|
>{\centering\arraybackslash}m{0.63cm}|
>{\centering\arraybackslash}m{0.63cm}|
>{\centering\arraybackslash}m{0.63cm}|
>{\centering\arraybackslash}m{0.63cm}|
>{\centering\arraybackslash}m{0.63cm}|
>{\centering\arraybackslash}m{0.63cm}|
>{\centering\arraybackslash}m{0.63cm}|
>{\centering\arraybackslash}m{0.63cm}|
>{\centering\arraybackslash}m{0.63cm}|
}

	\hline
        \centering\textbf{Contributions} 
        &\centering\textbf{This paper}
        &\centering\cite{2024:an:nearfieldsim}
        &\centering\cite{An:SIMHMIMO:JSAC:2023}
        &\centering\cite{Perovic:CL:2024}
        &\centering\cite{sim:isac:2024}
        &\centering\cite{madrl:ris}
        &\centering\cite{Shi:MCOM:2022}
        &\centering\cite{cite:Chien:TWC:2022}
        &\centering\cite{Kaixi:2021:China}
        &\centering\cite{cf:ioe:swipt:yang}
        &\centering\cite{hu:cfmimo:dl}
        &\centering\cite{Li:TC:2024}
        &\centering\cite{shi:arxiv:2024}
        &\centering\cite{shi:APUE:2025}
        \cr

        \hline

        SIM     
        &\centering\checkmark 
        &\centering\checkmark
        &\centering\checkmark
        &\centering\checkmark
        &\centering\checkmark
        &\centering\tikzxmark
        &\centering\tikzxmark
        &\centering\tikzxmark
        &\centering\tikzxmark
        &\centering\tikzxmark
        &\centering\checkmark 
        &\centering\checkmark 
        &\centering\checkmark
        &\centering\checkmark
        \cr        
        \hline

        CF-mMIMO         
        &\centering\checkmark
        &\centering\tikzxmark
        &\centering\tikzxmark
        &\centering\tikzxmark
        &\centering\tikzxmark
        &\centering\checkmark
        &\centering\checkmark
        &\centering\checkmark
        &\centering\tikzxmark
        &\centering\checkmark
        &\centering\checkmark
        &\centering\checkmark
        &\centering\checkmark
        &\centering\checkmark
        \cr
        \hline

        SWIPT         
        &\centering\checkmark
        &\centering\tikzxmark
        &\centering\tikzxmark
        &\centering\tikzxmark
        &\centering\tikzxmark
        &\centering\tikzxmark
        &\centering\checkmark
        &\centering\tikzxmark
        &\centering\checkmark
        &\centering\checkmark
        &\centering\tikzxmark
        &\centering\tikzxmark
        &\centering\tikzxmark
        &\centering\tikzxmark
        \cr
        \hline

        Statistical Design  
        &\centering\checkmark
        &\centering\tikzxmark
        &\centering\tikzxmark
        &\centering\tikzxmark
        &\centering\tikzxmark
        &\centering\tikzxmark
        &\centering\tikzxmark
        &\centering\checkmark
        &\centering\tikzxmark
        &\centering\tikzxmark
        &\centering\tikzxmark
        &\centering\tikzxmark
        &\centering\checkmark
        &\centering\tikzxmark
        \cr
        \hline 

        CTCE DRL  
        &\centering\checkmark
        &\centering\tikzxmark
        &\centering\tikzxmark
        &\centering\tikzxmark
        &\centering\tikzxmark
        &\centering\checkmark
        &\centering\tikzxmark
        &\centering\tikzxmark
        &\centering\tikzxmark
        &\centering\tikzxmark
        &\centering\tikzxmark
        &\centering\tikzxmark
        &\centering\tikzxmark
        &\centering\tikzxmark
        \cr
        \hline

        CTDE DRL   
        &\centering\checkmark
        &\centering\tikzxmark
        &\centering\tikzxmark
        &\centering\tikzxmark
        &\centering\tikzxmark
        &\centering\tikzxmark
        &\centering\tikzxmark
        &\centering\tikzxmark
        &\centering\tikzxmark
        &\centering\tikzxmark
        &\centering\tikzxmark
        &\centering\tikzxmark
        &\centering\tikzxmark
        &\centering\tikzxmark
        \cr
        \hline
        
\end{tabular}
\vspace{-1.0em}
\label{tab:contribution}
\end{table*}
\vspace{-0.5em}
\subsection{Literature Review}
\subsubsection{RIS-assisted CF-mMIMO}
Shi~\textit{et al.}\cite{Shi:MCOM:2022} provided a comprehensive overview of opportunities and challenges in this domain and proposed efficient resource allocation algorithms for wireless energy transfer. Chien~\textit{et al.}~\cite{cite:Chien:TWC:2022} analytically investigated CF-mMIMO systems with RISs for transmission to a zone of receivers and proposed a low-complexity sub-optimal phase shift (PS) configuration to minimize NMSE, thereby enhancing SE. Yang~\ettall~\cite{Kaixi:2021:China} applied alternative optimization (AO) to maximize the max sum-rate (MSR) metric in RIS-aided CF-mMIMO SWIPT.
Yang~\ettall~\cite{cf:ioe:swipt:yang} employed particle-swarm-based optimization to jointly maximize sum-HE and SE metrics via PS optimization.

\subsubsection{SIM-assisted CF-mMIMO}
Li~\ettall~\cite{hu:cfmimo:dl} maximized SE metrics by optimizing SIM's PS using the AO technique. 
Li~\ettall~\cite{Li:TC:2024} studied the holographic cell-free systems in the presence of SIMs with hardware impairments.
Shi~\ettall~\cite{shi:arxiv:2024} investigated two-stage pilot allocation and MSR sub-problems using an AO method in uplink SIM-aided CF-mMIMO system.
Shi~\ettall~\cite{shi:APUE:2025} proposed a three-phase scheme that combines heuristic search, AO, and projected gradient descent (GD) for interference minimization, AP power transmission, and SIM's PSs for the MSR problem under instantaneous CSI-based design.

\subsubsection{SIMs versus RISs}
In recent years, SIM-related studies have progressed toward physics-based investigations that explicitly model electromagnetic effects, such as inter-layer reflections and mutual coupling, thereby revealing the physical mechanisms that impact the passive beamforming capability of practical SIMs. Li~\ettall~\cite{Li:TC:2024} investigated the substantial impact of SIMs' physical properties, including the number of metasurfaces per SIM, inter-layer distance, hardware quality factors onto the achievable rates
Nerini~\ettall~\cite{Nerini:CL:2024} proposed a physically consistent SIM-assisted single-input single-output transmission to illustrate the effect of mutual coupling on the SIM performance. 
Yahya~\ettall~\cite{Yahya_coupling_sim} proposed a tractable T-parameter-based model that simplifies the coupling-aware analysis of cascaded metasurfaces. Li~\ettall~\cite{Li_LEO_2025} investigated SIM-based beamforming for multi-altitude satellite networks, revealing that incorporating mutual coupling into the beamforming design can effectively compensate for the reduced throughput. The numerical findings demonstrated that it is crucial to account for physical and hardware impairments in SIM-assisted wireless systems.

\vspace{-0.5em}
\subsection{Research Gap and Summarized Contributions}
It is well-established that SIMs hold great promise in extending coverage and mitigating blockage effects, but their integration into CF-mMIMO systems with SWIPT remains underexplored. In addition, existing studies primarily rely on traditional optimization frameworks such as AO, GD, and successive convex approximation (SCA). Although these methods offer advantages in analytical tractability, they share a common limitation of poor scalability. Specifically, AO is prone to local optima, SCA depends on complex convex approximations, and all entail considerable computational overhead~\cite{optimization:survey:2023}.

Deep reinforcement learning (DRL) offers a compelling alternative by inherently capturing non-convex problem characteristics without requiring handcrafted convexification. Conventional DRL frameworks typically employ a \textit{centralized training and centralized execution} (CTCE) strategy, in which a single agent, equipped with a deep neural network (DNN) actor handles all decision-making for the entire network \cite{HuaTcom2025, cite:thien_iotj_2023, 2023:tuan:maddpg}. In this strategy, the single agent observes a global state, computes a global joint action, and optimizes its policy through a centralized backpropagation process. Once trained, this single agent must also execute the global policy in real-time. Figure 1 illustrates a SIM-assisted CF-mMIMO SWIPT system operating under the CTCE framework, where a DNN-equipped central processing unit (CPU) is responsible for training and subsequently determining all signal processing, power allocation, active and passive beamforming for all APs. Consequently, as the CF-mMIMO networks scale up, the single-processing CPU faces an exponentially growing dimensionality of optimizing variables, leading to increased computational complexity, communication and fronthaul overhead. These factors constitute critical bottlenecks for learning large-scale CF-mMIMO deployments.

To overcome the scalability limitations of CTCE, we hereafter propose a centralized training, decentralized execution (CTDE) framework tailored for SIM-assisted CF-mMIMO SWIPT systems, as shown in Fig. 2. Compared to the CTCE strategy, CTDE equips each AP with its own DNN-based actor and decouples the learning and execution phases. This elaborates \textit{intelligent distributing optimization}, where each AP makes optimization decisions based on its local observations and a lower-dimensional action space, instead of relying on a single centralized policy that must jointly output all AP decisions. The distributed learning strategy significantly reduces the input–output dimensionality of each actor network, which in turn reduces the DNN 's size and per-step computational cost, making both training and inference more scalable in large-scale CF-mMIMO deployments.
At the same time, CTDE improves the learning behavior because it mitigates the \textit{curse of dimensionality} inherent in fully centralized DRL as the policy gradient variance and critic estimation errors typically grow when the state and action spaces become excessively high-dimensional. The CTDE strategy tends to yield more stable updates and effective learning convergence, while ensuring high scalability and low fronthaul latency for large-scale networks, efficiently handling the growing complexity of CF-mMIMO SWIPT systems.

\begin{figure*}[t]
    \centering
    \begin{minipage}[t]{0.45\textwidth}
        \centering
        \includegraphics[width=\textwidth]
        {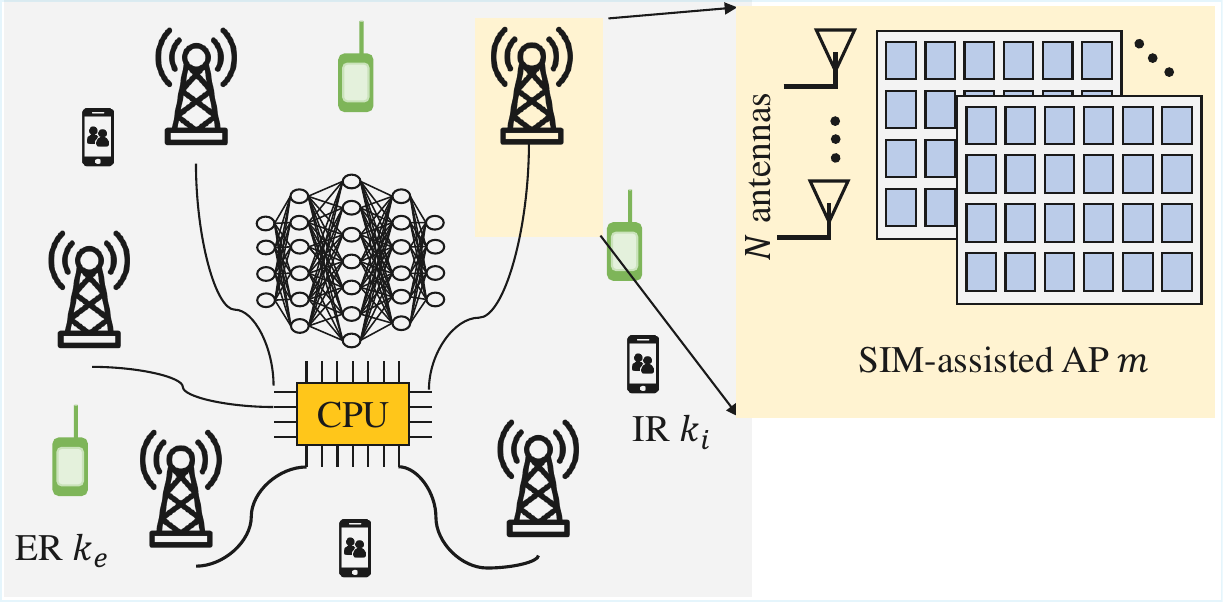} 
        \vspace{-0.6em}
        \caption{\small The CTCE CF-mMIMO-assisted SWIPT using SIMs.}
        \label{fig:system_model_a}
    \end{minipage}
    \hfill 
    \begin{minipage}[t]{0.45\textwidth}
        \centering
        \includegraphics[width=\textwidth]{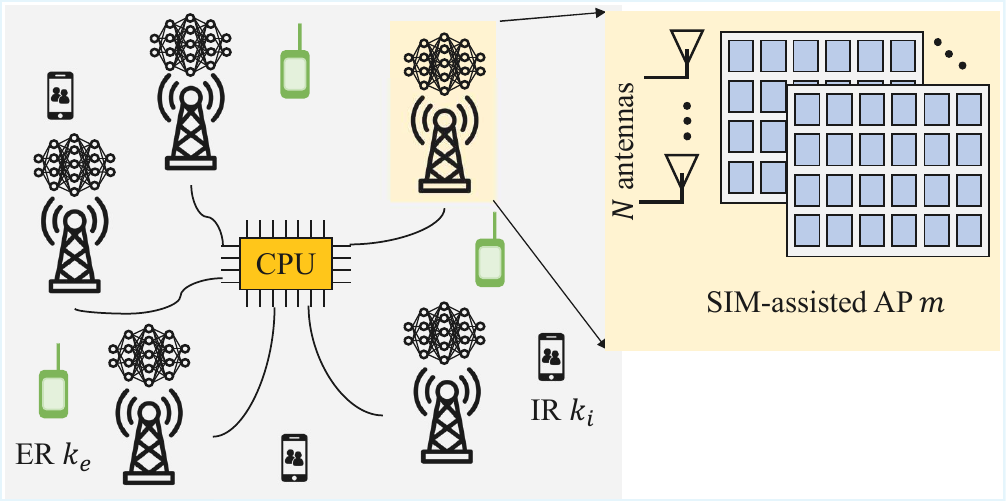} 
        \vspace{-0.6em} 
        \caption{\small The CTDE CF-mMIMO-assisted SWIPT using SIMs.}
        \label{fig:system_model_b}
    \end{minipage}
\vspace{-0.5em} 
\end{figure*}

To this end, we focus on a CF-mMIMO SWIPT system serving two distinct user groups: information-decoding receivers (IRs) and energy-harvesting receivers (ERs). To optimize resource utilization, we adopt the methodology from~\cite{Mohammadi:GC:2023}, categorizing APs into information-APs (I-APs) and energy-APs (E-APs). Each AP's signal transmission is enhanced by integrating a SIM within its antenna radome. We employ the protective partial zero-forcing (PPZF) precoder to safeguard transmissions to IRs from interference by applying partial ZF at I-APs and protective maximum ratio transmission (PMRT) at E-APs~\cite{Hua:WCNC:2024}. The primary contribution of this work is enabling autonomous decision-making through a CTDE DRL framework, which leverages DNNs at both the individual APs and a central CPU.
Our main contributions can be summarized as follows:
\begin{itemize}
    \item We explore a statistical CSI-based SIM-assisted CF-mMIMO SWIPT system, which accounts for channel estimation errors, pilot contamination (PC), and non-linear energy harvesting (NLEH) effects. Under these conditions, we derive closed-form expressions for the ergodic SE of the IRs and the average HE of the ERs. Using these expressions, we formulate a mixed-integer non-convex optimization problem to jointly determine the AP operation mode, power allocation, and SIMs' PS configuration, aiming to maximize the average total HE at the ERs, while meeting minimum SWIPT quality-of-service (QoS) requirements.
    \item The optimization problem is reformulated as a Markovian environment, enabling the APs to interact with it and jointly optimize solutions via a reward-driven feedback mechanism. A reward normalization step is then introduced to enhance learning stability, reduce hyperparameter sensitivity, and lower policy gradient variance, ultimately improving DRL convergence.
    \item To serve as a robust comparison benchmark, we develop a two-stage convex-based optimization framework. In the first stage, we heuristically optimize the SIMs' PS. In the second stage, we employ a convex-based method to jointly address the remaining subproblem, which entails APs mode selection and transmit power allocation.
    \item Through extensive numerical simulations, we validate the effectiveness of integrating SIMs into the CF-mMIMO SWIPT framework. A comprehensive performance comparison is carried out among various CTDE and CTCE learning strategies, and the convex-based benchmark, highlighting trade-offs in solution robustness, computational complexity, and real-time adaptability.
\end{itemize}

A comparison of our contributions against the SIM-oriented literature is presented in Table~\ref{tab:contribution}.


\textit{Notation:} 
We use bold lower (capital) case letters to denote vectors (matrices). The superscript $(\cdot)^\dag$ stands for the Hermitian transpose; $\mathbf{I}_N$ denotes an $N\times N$ identity matrix; $\boldsymbol{0}_{N \times K}$ denotes the matrix with all elements zero; $[\qA]_{i,j}$ denotes the $(i,j)$th entry of $\qA$; $\trace(\cdot)$ and $\diag(\cdot)$ denote the trace and diagonal operators, respectively; $\Vert\cdot\Vert_{\mathrm{F}}$ denotes the Frobenius norm; $\text{mod}(\cdot,\cdot)$ is the modulus operation; $ \lceil \cdot \rceil$ denotes the ceil function; $\bar{\qx}$ and $\tilde{\qx}$ denote the  line-of-sight (LoS) and  non-line-of-sight (NLoS) terms of a complex-valued Gaussian distributed vector $\qx$; a circularly symmetric complex Gaussian variable with variance $\sigma^2$ is denoted by $\mathcal{CN}(0,\sigma^2)$. Finally, $\mathbb{E}\{\cdot\}$ denotes statistical expectation.

\vspace{-0.4em}
\section{System Model}~\label{sec:Sysmodel}
Consider a SIM-assisted CF-mMIMO SWIPT system where $M$ APs simultaneously serve $\Kii$ single-antenna IRs and $\Kee$ single-antenna ERs over the same frequency band. We assume perfect synchronization among the distributed APs.\footnote{Synchronization among the APs is important from a practical point of view. According to \cite{synchronization}, the CPU can estimate residual per-AP synchronization errors using uplink pilots. Over one coherence block of duration $T_C$, carrier frequency offset $\Delta f_m$ induces a phase rotation $\Delta \phi_m \approx 2 \pi \Delta f_m T_C$. Subsequently, a practical verification step is to estimate $\Delta f_m$ and determine that the phase misalignment over $T_C$ is within a small tolerance such that the performance gap to the perfect synchronization bound is negligible. This is feasible since $\Delta f_m$ is pilot-aided and can be estimated with negligible overhead beyond the existing UL training.} For notational simplicity, we define the sets $\K$, $\Kii$, $\Kee$, $\M$, $\mathcal{L}$, and $\mathcal{S}$ to respectively collect the indices of the total receivers, IRs, ERs, APs, SIMs, and reflective elements per metasurface. We model AP $m$ with a uniform linear array (ULA) of $N$ active antennas that uniformly illuminates its attached SIM as in~\cite{sim:isac:2024, An:Arxiv:2023}.
Each SIM comprises $L$ metasurfaces, where each metasurface comprises $S$ reflective elements.
We denote the attached $L$-layer SIM at AP-$m$ as $\boldsymbol{\Phi}^{m \ell} \in \C^{S \times S} \triangleq \diag( e^{j\theta^{m \ell}_{1}}, \ldots, e^{j\theta^{m \ell}_{s}}, \ldots, e^{j\theta^{m \ell}_{S}} )$, where  $\theta^{m \ell}_{s} \in [0, 2 \pi]$ is the PS element.

\begin{Remark}
    We denote a coupling-aware PS of metasurface $(m, \ell)$ as $(m, \ell)$ as $\tilde{\boldsymbol{\Phi}}^{m \ell}\triangleq \boldsymbol{\Phi}^{m \ell}\big( \qI_{S} - \boldsymbol{\Phi}^{m \ell} \qS^{m, \ell} \big)^{-1}$ captures multiport properties including mutual coupling and multiple reflections \cite{coupling_aware_ps}. By inserting a carefully designed impedance matching or decoupling network between the SIM elements and their loads, the off-diagonal coupling terms can be canceled \cite{nossek:circuittheory, semmlerTWC2025}. This technique makes the overall response of the compensated SIM approximately diagonal. For analytical tractability, we assume negligible mutual coupling in the propagation manipulated through the SIMs, mathematically speaking $\qS^{m, \ell} \approx \boldsymbol{0}_{S\times S}$, thus $\tilde{\boldsymbol{\Phi}}^{m \ell}= \boldsymbol{\Phi}^{m, \ell}, \forall m \in \M, \ell \in L$. This serves as a model for a coupling-compensated SIM, allowing system performance bounds to be derived.
\end{Remark}

\vspace{-0.8em}
\subsection{Channel Model}
We assume a block fading channel that remains invariant and frequency-flat during each coherence interval $\tau_{c}$, including $\tau$ samples for UL channel estimation training and the rest for downlink (DL) transmission. For a general interpretation, let $k \in \K$ be the index of the receivers and $\mathtt{i}\in\{\mathtt{I}, \mathtt{E}\}$ be the index corresponding to wireless information transmission (WIT) and WPT operations. 
In each coherence interval, the instantaneous channel from AP $m$ to receiver $k$ is mathematically expressed as $\gmk \triangleq \FmRIS^{\dag} \zmk$, where $\FmRIS \in \C^{S \times N}$ and $\zmk \in \C^{S \times 1}$ denote the propagations from the ULA to the SIM and from the last SIM layer of AP $m$ to receiver $k$, respectively. Mathematically speaking:
\vspace{-0.5em}
\begin{equation}~\label{eq:Fm}
\FmRIS \triangleq 
        \boldsymbol{\Phi}^{m L} \qH^{m L} 
        \cdots 
        \boldsymbol{\Phi}^{m \ell} \qH^{m \ell}
        \cdots 
        \boldsymbol{\PHI}^{m 1} \qH^{m 1}, 
\end{equation}
where $\qH^{m 1} \in \C^{S \times N}$ and $\qH^{m \ell} \in \C^{S \times S}$ are the EM-wave propagation matrices from AP $m$ to the first SIM layer and from the $(\ell - 1)$-th layer to the $\ell$-th layer, $\forall \ell > 1$, respectively. According to Rayleigh-Sommerfeld diffraction theory~\cite{An:SIMHMIMO:JSAC:2023}, the coefficient from the $\breve{s}$-th element of the $(\ell - 1)$-th layer to the $s$-th element of the $\ell$-th layer is formulated as
\vspace{-0.2em}
\begin{equation}~\label{eq:rayleigh_sommerfeld_coeffi}
[\qH^{m \ell}]_{s,\breve{s}} \!=\! \frac{\lambda^2 \cos(\chi^{m \ell}_{s,\breve{s}})}{4 d^{\ell}_{s,\breve{s}}}\Big(\frac{1}{2\pi d^{\ell}_{s,\breve{s}}} \!-\! \frac{j}{\lambda} \Big) \exp\Big( \frac{j 2 \pi d^{\ell}_{s,\breve{s}}}{\lambda} \Big),
\end{equation}
where $\lambda$ is the wavelength, $\chi^{m \ell}_{s,\breve{s}}$ is the angle between the propagation and normal directions of the $(\ell - 1)$-th layer, while $d^{\ell}_{s,\breve{s}}$ is the geometric distance from the $\breve{s}$-th element of the $(\ell - 1)$-th layer to the $s$-th element of the $\ell$-th layer, i.e.,
\vspace{-0.3em}
\begin{equation}~\label{eq:3D_distance}
d^{\ell}_{s,\breve{s}} = \sqrt{\Big( d_{\mathtt{PS}} \sqrt{(s_{z} - \breve{s}_{z} )^2 + (s_{y} - \breve{s}_{y})^2} \Big)^{2} + d_{\mathtt{SIM}}^{2} },
\end{equation}
where $d_{\mathtt{PS}}$ and $d_{\mathtt{SIM}}$ are the spacings between two adjacent elements and metasurfaces, respectively; the indices of the elements $s_{z}$ and $s_{y}$ along the $z$- and $y$-axes can be computed as $s_{z} \triangleq \lceil s/ \sqrt{S} \rceil$ and $s_{y} = \mod(s-1, \sqrt{S}) + 1$, respectively.

We note that the propagation coefficient from the $n$-th transmit antenna to the $s$-th element of the first SIM layer $[\qH^{m, 1}]_{s,n}$ is formulated by~\eqref{eq:rayleigh_sommerfeld_coeffi}. The Ricean instantaneous channel from the last SIM layer of AP $m$ to receiver $k$ is modeled as
\vspace{-0.3em}
\begin{align}
    \zmk \triangleq \sqrt{\barbetamk}\big(\sqrt{\kappa }\zmklos + \zmknlos \big),
\end{align}
where $\barbetamk \triangleq \betamk/(1+\kappa)$, with $\betamk$ representing the large-scale fading coefficient (LSFC),  $\kappa$ is the Ricean factor, $\zmknlos$ is the NLoS component, 
whose $s$-th element is distributed as 
$[\zmknlos]_{s} \sim \CN(0,1), \forall m \in \M \text{, } \forall k \in \K$. In addition,  $\zmklos$  is the LoS component whose $s$-th element  is given by
\vspace{-0.5em}
\begin{align}
    \big[ \zmklos \big]_{\!s} 
    &\!=\! \exp \!\Big( 
    \!
    \zeta s_{z} \sin{\big(\chi^{L}_{s,k} \big)} 
    \!+\!
    \zeta s_{y} \sin\big(\varepsilon^{L}_{s,k}\big) \cos\big(\chi^{L}_{s,k}\big) \Big),
\end{align}
where $\zeta \triangleq j 2 \pi d_{\mathtt{PS}} / \lambda$, $\sin\big(\varepsilon^{L}_{s,k}\big) \cos\big(\chi^{L}_{s,k}\big) = (n_{z} \!-\! k_{z})/d^{L}_{n,k}$, $\sin{\big(\chi^{L}_{s,k} \big)} = (\vert n_{z} - k_{z} \vert)/\big(d^{L}_{n,k} \big)$, 
where $d^{L}_{n,k}$ is the geometric distance from receiver $k$ to  element $n$ of the last SIM layer, while $\chi^{L}_{n,k}$ and $\varepsilon^{L}_{n,k}$ are the elevation and azimuth angle between the propagation direction and the $Oxy$ and $Oyz$ surfaces, respectively.
To facilitate subsequent derivations, we present the channel statistics in the following remark.
\begin{Remark}~\label{remark:2ndmoment_gmk}
By leveraging the second-order moment of the channel components, we obtain
\vspace{-0.5em}
\begin{equation}
    \Ex\{\Vert \gmk \Vert^{2} \} \!\!=\!\! \kappa \barbetamk \trace\big( \FmRIS \FmRIS^{\dag} (\zmklos)^{\dag} \zmklos  \big)
    \!\!+\!\! \barbetamk \trace\big( \FmRIS \FmRIS^{\dag} \big),\nonumber
\end{equation}
while $\gmk$ is distributed as $\gmk \sim \CN \big(\gmkbar, \barbetamk \FmRIS^\dag \FmRIS \big)$, where $\gmkbar \triangleq \sqrt{\barbetamk \kappa} \FmRIS^{\dag} \zmklos$~\cite[Lemma~1]{Hua:WCNC:2024}. 
\end{Remark}

\vspace{-0.9em}
\subsection{Uplink Channel Estimation}\label{phase:ULforCE}
In the training phase, all IRs and ERs transmit their pilot sequences of length $\tau$ symbols to the APs.
We consider the general case where shared pilot sequences are used and denote $\mathcal{P}_k \subset \K $ as the set of receivers $k'$, including $k$, that are assigned with the same pilot sequence as the receiver $k$. In this regard, $i_{k} \in \{1, \ldots, \tau \}$ denotes the index of the pilot sequence used by receiver $k$. The pilot sequences $\VARPHI_{i_k} \in \C^{\tau \times 1}$, with $\Vert \VARPHI_{i_k} \Vert^{2} = 1$, are mutually orthogonal so that $\VARPHI_{i_{k'}}^{\dag}\VARPHI_{i_k}=1 \text{ if } i_{k'} \in \mathcal{P}_k$ and $\VARPHI_{i_{k'}}^{\dag}\VARPHI_{i_k}=0,\text{ otherwise}$. Then, the signal received at AP $m$ can be expressed as
\vspace{-0.3em}
\begin{equation}~\label{eq:receivedpilotsequence}
    \qY_{p,m}^{\mathtt{i}} \!=\! \sqrt{\tau\snrul }\big( \gmk  \VARPHI_{i_k}^\dag  \!+\!\sum\nolimits_{k'\in \K \setminus k}\! \qg^{\mathtt{i}}_{mk'}  \VARPHI_{i_{k'}}^\dag\big) \!+\! \qN_{p,m},
\end{equation}
where $\snrul$ is the normalized UL signal-to-noise ratio (SNR), while $\qN_{p,m} \in \C^{N \times \tau}$ is the receiver noise matrix containing independent
and identically distributed (i.i.d.) $\CN(0,\Snn)$ random variables (RVs). Similar to \cite{Hua:WCNC:2024}, to estimate $\gmk$, we 
first project $\qY_{p,m}^{\mathtt{i}}$ onto the $i_{k}$-th pilot sequence, and then apply the linear minimum mean-squared error (MMSE) estimation technique. The corresponding channel estimate is
\vspace{-0.3em}
\begin{align}
    \hgmk 
    & \!=\! \Ex\{ \gmk \} 
    \!+\! 
     \qC_{\gmk \ymk }\big(\qC_{\ymk \ymk }\big)^{\!\!-1}  \big(\ymk \!-\! \Ex\{ \ymk \}\big)
     \nonumber\\
    &\!= \!\sqrt{\barbetamk \kappa} \FmRIS^{\dag} \zmklos +     \qA_{mk}
    \nonumber\\
    & \hspace{2em} \times
    \Big(\!  \sqrt{\tau \snrul}
        \sum\nolimits_{k' \in \mathcal{P}_k } 
        \!\! { \sqrt{ \barbetamkp} \FmRIS^{\dag} \zmkpnlos + \Tilde{\qn}_{p,mk} } 
        \!\Big),
\end{align}
where $\ymk\!=\!\!\qY_{p,m}^{\mathtt{i}}\VARPHI_{i_k}$, $\Tilde{\qn}_{p,mk} \!\triangleq\! \qN_{p,m} \VARPHI_{i_k} \!\sim \!\CN(\boldsymbol{0},\Snn \qI_{N})$, and $$\qA_{mk}\!=\!\! \Big(\!\sqrt{\tau \snrul} \!\barbetamk \FmRIS^{\dag} \FmRIS \Big) \! \Big( \!\tau \snrul \!\sum\nolimits_{k' \in \mathcal{P}_k}{\!\! \barbetamkp }  \FmRIS^{\dag} \FmRIS \!+ \Snn \qI_{N}\Big)^{\!-1}.$$ Thus, the statistical distribution of $\hgmk$ follows $\hgmk \sim \CN \big( \gmkbar, \Covhatgmk \big)$, where $\Covhatgmk \triangleq \sqrt{\tau \snrul} \barbetamk \FmRIS^{\dag} \FmRIS \qA_{mk}$. Moreover, we obtain
\vspace{-0.3em}
\begin{equation}
    \Ex\{\Vert \hgmk \Vert^2 \} \triangleq  \kappa \barbetamk \trace(\FmRIS \FmRIS^\dag \zmklos (\zmklos)^\dag ) + \gamma^{\mathtt{i}}_{mk},
\end{equation}
where
\vspace{-0.7em}
\begin{equation}
    \gamma^{\mathtt{i}}_{mk} = \frac{\tau \snrul (\barbetamk)^2 \big(\trace(\FmRIS \FmRIS^{\dag})\big)^2 }
    {\tau \snrul \sum_{k' \in \mathcal{P}_k}{ \barbetamkp } \trace(\FmRIS \FmRIS^{\dag}) + N \Snn }.
\end{equation}
Subsequently, the expectation of the norm square of the estimate error $\tilgmk \triangleq \gmk - \hgmk$ is $\Ex \{ \Vert \tilgmk \Vert^2 \}
    = \barbetamk \trace(\FmRIS \FmRIS^{\dag}) - \gamma^{\mathtt{i}}_{mk}$.\footnote{Pilot contamination interacts with SIM phase-shift optimization through the channel estimation quality. Particularly, the MMSE estimation variance term depends on the SIM configuration via the quadratic form $\trace (\qF_m \qF_m^{\dag})$. As a result, the receivers sharing the same pilot contribute to the denominator of the term $\gamma_{mk}^{\mathtt{i}}$, scaled by their large-scale fading and the SIMs phase-shift. Hence, optimizing the SIM phase shifts not only shapes the downlink beamforming gain, but also implicitly alters the effective pilot contamination level during the uplink training phase.}

\subsection{Downlink Wireless Information and Power Transmission}
\vspace{-0.1em}
We define the binary variable $a_m$ as the operation indicator, where $a_m = 1$ ($a_m = 0$) if the $m$-th AP operates as an I-AP (E-AP). Let $\xik$ ($\xej$) denote the information (energy) symbol transmitted to  IR $k_i$ (ER $k_e$)  that satisfies $\Ex\{ \vert \xik \vert^2 \} = \Ex\{ \vert \xej \vert^2 \} =1$ and $\etamkI$ ($\etamjE$) denote the power control coefficients associated with AP $m$ and IR $k_i$ (ER $k_e$). Hence, the  signal transmitted from the $m$-th AP can be expressed as
\vspace{-0.2em}
\begin{align}~\label{eq:x_m}
    \qx_{m}
    &= \sum\nolimits_{k_{i}\in\Kii}\sqrt{a_m\rho_{d}\etamkI} \wimk \xik \nonumber \\
    &+  \sum\nolimits_{k_e\in\Kee} \sqrt{(1-a_m)\rho_{d}\etamjE}\wemj \xej,
\end{align}
where $\snrdl = \tilde{\rho}_{d}/\Snn$ is the normalized DL SNR; $\wimk \in \C^{N\times 1}$ ($\wemj\in \C^{N\times 1}$) represents the precoding vector for IR $k_i$ (ER $k_e$) with $\Ex\big\{\big\Vert\wimk\big\Vert^2\big\}=\Ex\big\{\big\Vert\wemj\big\Vert^2\big\}=1$; $\etamkI$ and $\etamjE$ are the power control coefficients at AP $m$, chosen to satisfy the power transmission constraint
\vspace{-0.2em}
\begin{align}~\label{eq:transmitpowerconstraint}
    & a_m\Ex\big\{\big\Vert \qx_{\mathtt{I},m}\big\Vert^2\big\}+ (1-a_m)\Ex\big\{\big\Vert \qx_{\mathtt{E},m}\big\Vert^2\big\}\leq \rho_{d}.
\end{align}
Then, the IR and ER respectively receive
\vspace{-0.2em}
\begin{subequations}
 \begin{align}
    \yik 
    &\!=\!  
    \sum\nolimits_{k'_i\in \Kii } \sum\nolimits_{m \in \M } \sqrt{a_m \rho_d\etamkpI} (\gmkiu)^\dag \wimkp \xikp \nonumber \\
    &\hspace{-2em}+\! \sum\nolimits_{k_e\in\Kee}  \! \sum\nolimits_{m \in \M } \!\! \sqrt{ \!(1\!-\!a_m) \rho_d\etamjE} \!(\gmkiu)^{\!\dag} \wemj \xej 
    \!\!+\! \!
    n_{d,k_{i}},\nonumber\\
    \yej
    &\!\!=\!  \sum\nolimits_{k'_e\in \Kee }\!  \sum\nolimits_{m \in \M}\!\!
    \sqrt{ \!(1\!-\!a_m) \rho_d\etamjpE} (\gmjue)^\dag \wemjp \xejp   \nonumber \\
    &\hspace{-2em}+\!  
    \sum\nolimits_{k_i \in\Kii }\!
    \sum\nolimits_{m\in\M }\!\! 
    \sqrt{a_m \rho_d\etamkI} (\gmjue)^\dag \wimk \xik \! +\! n_{d,k_{e}}\!, \nonumber
\end{align}
\end{subequations}
respectively, where $n_{d,k_{i}},n_{d,k_{e}} \!\sim \! \CN(0,1)$ are the additive noise terms.
\vspace{-0.8em}
\subsection{Protective Partial Zero Forcing Precoding} 
According to~\cite{Mohammadi:GC:2023, cite:pzf_pmrt_2020}, we employ the PPZF precoding scheme, which strategically combines distinct beamforming techniques for the two AP modes. At the I-APs, zero-forcing (ZF) is used to eliminate inter-user interference among the information receivers. At the E-APs, the goal is to use PMRT, as it is optimal for wireless power transfer, especially in large-antenna systems.
However, conventional MRT from the E-APs would cause significant interference to the IRs. To address this, PPZF implements a crucial modification: the PMRT precoder is projected into the orthogonal complement of the IRs' collective channel matrix. This modified scheme forces the energy signals into the null space of the IRs. As a result, if the channel estimation is perfect, the information transmission to IRs is fully protected from the energy signals directed at the ERs. Let $\Ghmu = \big\{ \hgmki \big\} \in \C^{N\times \Kii}$ and $\Ghms= \big\{ \hgmkue \big\} \in \C^{N\times \Kee}$ be  the matrices of the estimated channels between AP~$m$ and all IUs, and all EUs, respectively.
By invoking~\cite{Hua:WCNC:2024}, the ZF and PMRT precoders at AP $m$ for IR $k_i$ and ER $k_e$ are given by
\vspace{-0.3em}
\begin{subequations}~\label{eq:wemrtzf}
 \begin{align}
    \wimk^{\PZF} &= (\alphaPZFmki)^{-1} \Ghmu \Big(\big(\Ghmu\big)^\dag \Ghmu\Big)^{-1} \!\qe_{k_i}^{\mathtt{I}}
    ,~\label{eq:wipzf}
    \\
    \wemj^{\PMRT} &= (\alphaPMRTmke)^{-1} \qB_m\Ghms\qe_{k_e}^{\mathtt{I}}
    ,~\label{eq:wemrt}
\end{align}   
\end{subequations}
where $\qe_{k_i}^{\mathtt{I}}(\qe_{k_e}^{\mathtt{E}})$ denotes the $k_i(k_e)$-th column of $\qI_{K_i}(\qI_{K_e})$; $\qB_m$ denotes the projection matrix onto the orthogonal complement of $\Ghmu$ so that $\gmKiuH \qB_m =\boldsymbol{0}$. Thus, $\qB_m$ can be computed as
\vspace{-0.2em}
\begin{align}
  \qB_m  = \qI_{N}  - \Ghmu \Big( \big(\Ghmu\big)^\dag \Ghmu\Big)^{-1}  \big(\Ghmu\big)^\dag,
\end{align}
and $\Ex\{\qB_{m}\} \!=\! \Ex\{\qB_{m}^{\dag} \qB_{m}\} \!=\! \Ex\{\qB_{m} \qB_{m}^{\dag}\}=\qI_{N}, \forall m \in \M$, which was mathematically proven in~\cite{Mohamed_projection_2024}.
Note that in~\eqref{eq:wemrtzf}, $\alphaPZFmki$ and $\alphaPMRTmke$ are normalization factors, given by
\vspace{-0.3em}
\begin{subequations}
\begin{align*}
    \alpha_{m k_i}^{\PZF}&\triangleq\Big(\Ex \Big\{ \big\Vert \Ghmu \big(\big(\Ghmu\big)^{\dag} \Ghmu\big)^{-1} \qe_{k_i}^{\mathtt{I}}  \big\Vert^2 \Big\}  \Big)^{-\frac{1}{2}}
    \nonumber\\
    &=
    \Big(
    \Big[ \Ex \! \Big\{ \!
    \Big( \!
    \big(\Ghmu\big)^{\dag} \Ghmu
    \Big)^{\!\!-1}
    \Big\}
    \Big]_{k_i, k_i}
    \Big)^{\!\! -\frac{1}{2}},
    \\
    {\alpha_{m k_e}^{\PMRT}}&\triangleq\Big( \Ex \big\{ \big\Vert \qB_m\Ghms\qe_{k_e}^{\mathtt{E}} \big\Vert^2\big\} \Big)^{\!\! -\frac{1}{2}}
    \!\!=\!\!
    \Big( \!
    \Big[
    \Ex\Big\{\! \big(\Ghms\big)^{\!\dag} 
    \Ghms \!\Big\} \!
    \Big]_{k_e k_e}
    \Big)^{\!\! -\frac{1}{2}}.    
\end{align*}
\end{subequations}
Note that $\big(\Ghmu\big)^{\dag} \Ghmu$ follows a non-central Wishart distribution with the non-identical covariance matrices of the row vectors of $\Ghmu$.
Thus, it is challenging to obtain the exact closed-form expressions for $\alphaPZFmki$ and $\alpha_{m k_e}^{\PMRT}$.
\begin{table}[tbp]
\vspace{-0.3cm}
\caption{Summary of Notations}\label{parameter_notations}
\vspace{-0.2cm}
\centering
\scriptsize
    \begin{tabular}{|m{1.6cm}<\centering|m{6.5cm}<\centering|}
        \hline
        \textbf{Notation} & \textbf{Description} \\
        \hline
        $\tau_c, \tau$ & Coherence interval length and UL training duration (in symbols) \\\hline
        $\rho_d, \rho_u$ & Normalized downlink and UL SNR \\\hline
        $\Snn$ & AWGN power \\\hline
        $\VARPHI_{i_k}$ & Pilot sequence used by receiver $k$ \\\hline
        {$\gmk$} & {Actual channel vector for receiver $k$ at AP $m$} \\\hline
        {$\hgmk$} & {Estimated channel vector for receiver $k$ at AP $m$} \\\hline
        {$\tilgmk$} & {Estimated channel error for receiver $k$ at AP $m$} \\\hline
        {$\qF_{m}$} & {Cascade channel from AP $m$ through SIM}  \\\hline
        {$\mathbf{\Phi}_{m,\ell}$} & {Diagonal PS matrix of the $\ell$-th SIM layer at AP $m$}  \\\hline 
        {$\theta_{m,\ell}^s$} & {PS of the $s$-th element on the $\ell$-th layer at AP $m$} \\\hline
        {$\mathbf{H}^{m,\ell}$} & {Propagation matrix between layers $(\ell-1)$ and $\ell$} \\\hline
        {$\boldsymbol{\Theta}$} & {Set of all SIM PS configurations} \\\hline
        {$\zmk$} & {Channel vector from the last SIM layer of AP $m$ to receiver $k$} \\\hline
        {$\zmklos$ ($\zmknlos$)} & {Line-of-sight (non line-of-sight) component of the channel $\zmk$} \\\hline
        {$\betamk$} & {Large-scale fading coefficient between AP $m$ and receiver $k$} \\\hline
        {$\kappa$} & {Ricean fading factor} \\\hline
        {$\qx_{m}$} & {Transmitted signal from the AP $m$} \\\hline
        {$x_{\mathtt{I},m}$ ($x_{\mathtt{E},m}$)} & {Transmitted signal from I-AP (E-AP) $m$} \\\hline
        {$\yik$ ($\yej$)} & {Received signal at IR $k_i$ (ER $k_e$)} \\\hline
        {$n_{d,k_{i}}$ ($n_{d,k_{e}}$)} & {Additive noise at the IR $k_i$ (ER $k_e$)} \\\hline
        {$\wimk^{\PZF}$} & {ZF precoder at AP $m$ for IR $k_i$} \\\hline
        {$\wemj^{\PMRT}$} & {PMRT precoder for ER $k_e$} \\\hline 
        {$\xi, \chi, \phi, \Omega$} & {Circuit parameters for the NLEH model} \\\hline
        {$\SINRki$}  & {Effective SINR at IR $k_i$} \\\hline
        {$\SEk$}  & {Achievable spectral efficiency (SE) at IR $k_i$} \\\hline
        {$\mathcal{Q}_{k_e}$} & {Average harvested energy (HE) at ER $k_e$} \\\hline
        {$\Ex\big\{\mathcal{E}^{\mathrm{NL}}_{k_e}\big\}$}  & {Average non-linear harvested energy at ER $k_e$}  \\\hline
        {$\Gamma_{k_e}$}  & {Minimum EH QoS requirement for ER $k_e$} \\\hline
        {$\mathcal{S}_{k_i}$} & {Minimum SE QoS requirement for IR $k_i$}  \\\hline
        {$a_m$} & {Binary operation mode indicator for AP $m$} \\\hline
        {$\etamkI$ $(\etamjE)$} & {Power control coefficient for IR $k_i$ (ER $k_e$) at AP $m$} \\\hline
    \end{tabular}
    \vspace{-1em}
    \label{table:notation}
\end{table} 
To render the problem more tractable,  we deploy a tight approximation which relies on the law of large numbers (via the Tchebyshev's theorem) when $N \rightarrow \infty$. By using Tchebyshev's theorem~\cite{cramer1970}, we obtain
\vspace{-0.2em}
\begin{align}~\label{eq:16}
    \hspace{-0.6em}
    \frac{1}{N} 
    \Ex\bigg\{ \!
    \Big( \!
    \big(\Ghmu\big)^{\dag} \Ghmu 
    \Big)^{\!\!-1}
    \!\bigg\}
    \!-\!
    \frac{1}{N} 
    \Ex\bigg\{\!\! 
    \Big( \!
    \Ex \! \left\{ \!
    \big(\Ghmu\big)^{\dag} \Ghmu
    \right\}
    \Big)^{\!\!-1} 
    \!\bigg\}
    \!\xrightarrow[N \rightarrow \infty]{\text{a.s.}}  
    \!\boldsymbol{0}.
\end{align}
By using the above asymptotic result, we can obtain the following tight approximations
\vspace{-0.3em}
\begin{align}\label{eq:alphapzf:approx}
    \alphaPZFmki &\approx \big(\big[ (\hatWi)^{-1} \big]_{k_i, k_i} \big)^{\!-\frac{1}{2}},
\end{align}
where 
\vspace{-0.3em}
\begin{align}~\label{eq:alphapzf2:approx}
    \hatWi &= 
    \kappa (\qD^{\mathtt{I}}_{m})^{1/2} (\bar{\qZ}^{\mathtt{I}}_{m})^{\dag} \FmRIS \FmRIS^{\dag} \bar{\qZ}^{\mathtt{I}}_{m} (\qD^{\mathtt{I}}_{m})^{1/2} + \qD^{\mathtt{I}}_{m},
\end{align}
with $\bar{\qZ}^{\mathtt{I}}_{m}\triangleq \{ \zmkilos \}, \forall k_i \in \Kii$; $\qD^{\mathtt{I}}_{m} \in \R^{K_{\mathtt{I}} \times K_{\mathtt{I}}}$ is diagonal with $[\qD^{\mathtt{I}}_{m}]_{k_i, k_i} \triangleq \barbetamki$. The detailed proof of \eqref{eq:alphapzf:approx} is provided in Appendix~\ref{appendix:A1}.
Following a similar methodology to the derivation of \eqref{eq:alphapzf:approx}, we can derive $\alphaPMRTmke$ as 
\vspace{-0.2em}
\begin{align}~\label{eq:alphapmrt:approx}
    \alphaPMRTmke &\approx \big( \big[ \hatWe \big]_{k_e, k_e} \big)^{-\frac{1}{2}}, 
\end{align}
where
\vspace{-0.2em}
 \begin{align}
    \hatWe &= 
    \kappa (\qD^{\mathtt{E}}_{m})^{1/2} (\bar{\qZ}^{\mathtt{E}}_{m})^{\dag} \FmRIS \FmRIS^{\dag} \bar{\qZ}^{\mathtt{E}}_{m} (\qD^{\mathtt{E}}_{m})^{1/2} + \qD^{\mathtt{E}}_{m},
\end{align}
with $ \bar{\qZ}^{\mathtt{E}}_{m}\triangleq \{ \zmkelos \}, \forall k_e \in \Kee $; $\qD^{\mathtt{E}}_{m} \in \R^{K_{\mathtt{E}} \times K_{\mathtt{E}}}$ are diagonal matrices with $[\qD^{\mathtt{E}}_{m}]_{k_e, k_e} \triangleq \barbetamke$. 

To this end, the summary of notations used in the paper is listed in Table~\ref{parameter_notations}.

\begin{Remark}
    The approximation \eqref{eq:alphapzf:approx} is obtained from the asymptotic result \eqref{eq:16}, which is very tight in the massive MIMO regime where the number of antennas per AP is significantly larger than the number of served IRs, i.e., $N \gg K_{\mathtt{I}}$ \cite{cite:HienNgo:cf01:2017}. In this regime, the random $\Big( \! \big(\Ghmu\big)^{\dag} \Ghmu \Big)^{\!\!-1}$ matrix converges asymptotically to its deterministic mean.
\end{Remark}

\section{Performance Analysis and Problem Formulation}
\subsection{Ergodic Downlink Spectral Efficiency with PPZF}
Applying the use-and-then-forget technique~\cite{cite:HienNgo:cf01:2017}, the  signal received at IR~$k_i$ can be expressed as
\begin{align}~\label{eq:yi:hardening}
    \yik &=  \DSki  \xik +
    \BUki \xik 
         +\sum\nolimits_{k_{i}'\in \Kii \setminus k_{i}}
         \!\!
     \IUIki
     \xikp\nonumber\\
    &\hspace{1em}
    + \sum\nolimits_{k_{e}\in\Kee }
     \EUIki \xej + n_{k_i},~\forall k_i\in\Kii,
\end{align}
where $\DSki$, $\BUki$, $\IUIki$, and $\EUIki$ represent the desired signal, the beamforming gain uncertainty, the interference caused by the $k'_i$-th IR, and the interference caused by the $k_e$-th ER, respectively, given by
\begin{subequations}~\label{eq:yi:components}
    \begin{align}
        \DSki &\triangleq \sum\nolimits_{m\in\M} \sqrt{ a_m \snrdl \etamkI} \Ex \big\{(\gmkiu)^\dag \wimk^{\PZF} \big\},  
    \end{align}  
    \begin{align}
        \BUki &\triangleq \sum\nolimits_{m\in\M} \sqrt{a_m \snrdl \etamkI} \Bigl( (\gmkiu)^\dag \wimk^{\PZF} \nonumber \\
        &\hspace{1em}- \Ex \big\{(\gmkiu)^\dag \wimk^{\PZF} \big\} \Bigl)~\label{eq:component_BUk}, 
        \\
        \IUIki &\triangleq \sum\nolimits_{m\in\M} \sqrt{a_m \snrdl \etamkpI} (\gmkiu)^\dag \wimkp^{\PZF},~\label{eq:IU_interference}
        \\
        \EUIki &\triangleq \sum\nolimits_{m\in\M} \! \sqrt{(1 \!-\! a_m) \snrdl \etamjE} (\gmkiu)^\dag \wemj^{\PMRT}.  
    \end{align}  
\end{subequations}

Thus, the DL SE in [bit/s/Hz] for IR $k_i$ is
\begin{align}~\label{eq:SEk:Ex}
    \mathrm{SE}_k
      &=
      \Big(1\!- \!\frac{\tau}{\tau_c}\Big)
      \log_2
      \left(
       1\! + \SINRki\big(\boldsymbol{\Theta},\qa \ETAI, \ETAE\big)
     \right),
\end{align}
while $\SINRki$ is the effective SINR at IR $k_i$,  given by
\begin{align}~\label{eq:SINE:general}
   \SINRki= \frac{
             \vert  \DSki  \vert^2
             }
             {  
             \Ex\{\vert  \BUki  \vert^2\} 
            +\mathrm{IT}
            \!+\!  1},
\end{align}
where $\mathrm{IT}\!\triangleq \!\!\sum_{k_{i}'\in\Kii \setminus k_{i}}\! \Ex\{\vert \IUIki \vert^2\!\} \! + \! \sum_{k_e\in \Kee} \!\Ex \{\vert  \EUIki \vert^2\}$.

\begin{proposition}~\label{Theorem:SE:PPZF}
The SE of IR~$k_i$ is given by~\eqref{eq:SEk:Ex}, where
\begin{align}~\label{eq:numerical:SE:PPZF}
    \SINRki
    =
    \frac{
    \Big(
    \sum\nolimits_{m \in \M} { \alphaPZFmki \!
    \sqrt{a_m \snrdl \etamkI } 
     }
    \Big)^{2}
    }{
    \mathrm{PC}_{k_i} 
    \!+ \Omega_{k_i}    
    \!+\! 
    1
    },
\end{align}
while $ \mathrm{PC}_{k_i} = \sum\nolimits_{k_{i}'\in \mathcal{P}_{k} \setminus k_i } \bigg(\sum\nolimits_{m \in \M} \alphaPZFmkip \!\sqrt{a_m \snrdl \etamkpI }
    \bigg)^{\!\!2},$ and 
\vspace{-0.2em}
\begin{align}
    \Omega_{k_i} &= \sum\nolimits_{k_{i}'\in\Kii } \!\!
    \sum\nolimits_{m \in \M} \!
    a_m \snrdl \etamkpI \bar{e}_{m k_i} \nonumber\\
    &
    \!+\!
    \sum\nolimits_{k_{e}\in\Kee } \!
    \sum\nolimits_{m \in \M} \!
    (1-a_m)\snrdl \etamjE \bar{e}_{m k_i},
\end{align}
with $\bar{e}_{m k_i} = \barbetamki \trace(\FmRIS \FmRIS^{\dag} )\! -\! \gamma^{\mathtt{I}}_{m k_i}$.
\end{proposition}
\begin{proof}
See Appendix~\ref{appendix:B}.
\end{proof}

\vspace{-2em}
\subsection{Average Non-Linear Harvested Energy with PPZF}
We use the NLEH model~\cite{Hua:WCNC:2024,cite:HEModel:Boshkova,Mohammadi:TC:2024} to characterize the HE at ER $k_e$.
Subsequently, the total HE at ER $k_e \in \Kee$ is given by $\mathcal{E}^{\mathrm{NL}}_{k_e} \triangleq \big( \Lambda\big(\mathcal{E}_{k_e}\big) - \phi \Omega \big)(1-\Omega)^{-1}$, where $\phi$ is the maximum output DC power, $\Omega\triangleq1/(1 + \exp(\xi \chi) )$ is a constant to guarantee a zero input/output response, $\xi$ and $ \chi$ are circuit constants, and $\mathcal{E}_{k_e}$ denotes the received RF energy at ER $k_e$, $\forall k_e\in\Kee$. We obtain the energy at the output of the circuit as a logistic function of the received input RF signal at ER $k_e$, i.e. $\Lambda(\mathcal{E}_{k_e}) = \phi \big(1 + \exp\big(\!-\xi\big(\mathcal{E}_{k_e}\!-\! \chi\big)\!\big) \big)^{-1}$.
To this end, the average of the total HE at ER~$k_e$ is formulated as
\begin{align}~\label{eq:NLEH:av}
    \Ex\big\{\mathcal{E}^{\mathrm{NL}}_{k_e}\big\} 
    &= \frac{\Ex\big\{\Lambda\big(\mathcal{E}_{k_e}\big)\big\} - \phi \Omega }{1-\Omega}
    \approx
    \frac{\Lambda\big(\mathcal{Q}_{k_e}\big) - \phi \Omega }{1-\Omega}, 
\end{align}
where
\begin{align}~\label{eq:El_average}
    &\mathcal{Q}_{k_e}\! \!\triangleq\!
    (\tau_c\!-\!\tau)\Snn \snrdl
    \nonumber\\
    &\hspace{0.8em}\times\!\!
    \bigg(\!
    \sum\nolimits_{m\in\M}\!\!
    (1\!-\!a_m)\etamjE\! \Ex\Big\{\left\vert\!\gejueH \!\wemj^{\!\PMRT}\right\vert^2\!\Big\} 
    \nonumber\\
    &\hspace{0.8em}
    +
    \!
    \sum\nolimits_{k'_e \in\Kee \! \setminus k_e}\!
    \sum\nolimits_{m\in\M} \!\!
    {(1\!-\!a_m)\etamjpE} \Ex\Big\{\! \left\vert\gejueH\wemjp^{\PMRT}\right\vert^2\!\Big\}
    \nonumber\\
    &\hspace{0.8em} 
    +\!\!
    \sum\nolimits_{k_i\in\Kii}\!
    \sum\nolimits_{m\in\M }\!\!\!
    {a_m \etamkI}\Ex\Big\{\!\left\vert\big(\gejue\big)^{\!\dag}\wimk^{\PZF}\right\vert^2\!\Big\} 
    \!+\! \frac{1}{\snrdl}
    \bigg).
\end{align}
\begin{Remark}
    The approximation $\Ex\left\{\Lambda\left(\mathcal{E}_{k_e}\right)\right\} \approx\Lambda\left(\Ex\left\{\mathcal{E}_{k_e}\right\} \right)= \Lambda\left(Q_{k_e}\right)$ is applied on (27). The approximation relies on a first-order Taylor expansion of the NLEH logistic function around the mean received energy $\mathcal{Q}_{k_e}$, neglecting higher-order terms. This approximation is tight when the fluctuation of the received energy $\mathcal{E}_{k_e}$ is small \cite{Mohammadi:TC:2024}. In addition, the $\Lambda(\cdot)$ function has convexity and concavity ends of the logistic curve, which represents the low-power and saturation regimes, respectively. In extreme cases, the approximation may strictly serve as a lower bound (in the low-power regime) or an upper bound (in the saturation regime) for the actual average harvested energy.
\end{Remark}

\begin{proposition}~\label{Theorem:RF:MRT}
With PMRT precoding, the average HE at  ER $k_e \in \Kee$ is given by~\eqref{eq:NLEH:av}, where $\mathcal{Q}_{k_e}$ is given by
\begin{align}~\label{eq:El_average:PMRT}
    \mathcal{Q}_{k_e} &=
    (\tau_c - \tau)\Snn \snrdl
    \Big[ 1/ \snrdl +
    \sum\nolimits_{m \in \M} 
    (1-a_m) \etamjE  \bar{d}_{m k_e}
    \nonumber\\
    &\hspace{0em}
    \!+\!
    \sum\nolimits_{m \in \M} \!
    \sum\nolimits_{k'_e \in \mathcal{P}_k \setminus k_e} \!\!
    (1-a_m) \etamjpE  
    \bar{e}_{m, k_e k'_e}
    \nonumber\\
    &\hspace{0em}
    \!+\!
    \sum\nolimits_{m \in \M} \!
    \sum\nolimits_{k'_e \notin \mathcal{P}_k \setminus k_e}
    (1-a_m) \etamjpE \bar{f}_{m, k_e k'_e}
    \!\! 
    \nonumber\\
    &\hspace{0em}
    \!+\!
    \sum\nolimits_{m \in \M} \!
    \sum\nolimits_{k_i \in \Kii}
    a_m \etamkI (1/N) \bar{h}_{m, k_e}
    \Big],
\end{align}
while 
\vspace{-0.4em}
\begin{align}
    \bar{d}_{m k_e}
    &=
    \big(\alphaPMRTmke\big)^{\!2}
    \Big[
    \Big(
    \barbetamke \trace\big( \FmRIS \FmRIS^{\dag} \big) - \gameumj
    \Big)
    \nonumber\\
    &\hspace{0em} \times 
    \Big( 
    \kappa \barbetamke \trace\big(\FmRIS \FmRIS^{\dag} \zmkelos (\zmkelos)^{\dag}  \big) 
    \gameumj
    \Big)
    \!+\!
    \bar{a}_{m, k_e}
    \Big], 
    \nonumber\\
    \bar{e}_{m, k_e k'_e} &=
    \big(\alphaPMRTmke\big)^{\!2} 
    \nonumber\\
    &\hspace{0em}\times
    \Big[
    \bar{b}_{m, k_e k'_e}
    \!+\!
    \Big( 
    \kappa \barbetamkep \trace\big(\FmRIS \FmRIS^{\dag} \zmkeplos (\zmkeplos)^{\dag}  \big) 
    +
    \gameumjp
    \Big)
    \nonumber\\
    &\hspace{0em}\times
    \Big(
    \barbetamke \trace\big( \FmRIS \FmRIS^{\dag} \big) - \gameumj
    \Big)
    \Big], \forall k'_e \in \mathcal{P}_{k} \!\setminus\! k_e, 
    \nonumber\\
    \bar{f}_{m, k_e k'_e}
    &=
    \big(\alphaPMRTmke\big)^{\!2}
    \nonumber\\
    &\times
    \Big[
    \Big(\!
    \kappa \barbetamke \trace\big(\FmRIS \FmRIS^{\dag} \zmkelos (\zmkelos)^{\dag} \big)
    \!+\! \barbetamke \trace\big(\FmRIS \FmRIS^{\dag} \big)
    \Big)
    \nonumber\\
    &\hspace{0em} \times
    \Big(
    \kappa \barbetamkep \trace\big(\FmRIS \FmRIS^{\dag} \zmkeplos (\zmkeplos)^{\dag}  \big)
    + 
    \gameumjp
    \Big)
    \Big], 
    \nonumber\\
    &\hspace{0.5em}
    \forall k'_e \notin \mathcal{P}_{k} \!\setminus\! k_e,
    \text{ and}
    \nonumber\\
    \bar{h}_{m, k_e}
    &= 
    \big(
    \kappa \barbetamke \trace\big( \FmRIS \FmRIS^{\dag} \zmkelos (\zmkelos)^{\dag}  \big)
    \!+\! \barbetamke \trace\big( \FmRIS \FmRIS^{\dag} \big)
    \big),
    \nonumber
\end{align}
where
\vspace{-0.3em}
\begin{align*}
    \bar{a}_{m k_e} 
    &=\!
    \big\vert \trace\big(\Covhatgmke \big) \big\vert^{2}\! +\! \big\vert \trace\big(\Covhatgmke \Covhatgmke \big) \big\vert
    \!+\!
    \big\vert (\gmkebar)^\dag \gmkebar \big\vert^2
    \nonumber\\
    &+
    2 (\gmkebar)^\dag \gmkebar \trace\big(\Covhatgmke \big)
    +
    2 (\gmkebar)^\dag \Covhatgmke \gmkebar,
    \\
        \bar{b}_{m, k_e k'_e} &=
        \vert (\gmkebar)^{\dag} \gmkepbar \vert^2 
        \nonumber\\
        &\hspace{-2em}+
        2 \big( (\barbetamkep)^2 / (\barbetamke)^2 \big)(\gmkebar)^{\dag} \gmkpbar \trace\big(\Covhatgmk \big)
        \nonumber\\
        &\hspace{-2em}+
        \big( (\barbetamkep)^4 / (\barbetamke)^4 \big) \Big( \big\vert \trace\big(  \Covhatgmke \big) \big\vert^{2} 
        + \trace\big(\Covhatgmke \Covhatgmke \big) \Big)
        \nonumber\\
        &\hspace{-2em}+
        2 \big((\barbetamkep)^2 / (\barbetamke)^2 \big)(\gmkebar)^{\dag} \Covhatgmke \gmkepbar, \forall m \in \M.       
    \end{align*}
\begin{proof}
See Appendix~\ref{appendix:C}.
\end{proof}
\end{proposition}

\subsection{Problem Formulation}~\label{ref:ProblemFormulation}
We formulate an optimization problem that maximizes the average sum-HE, subject to the minimum HE QoS $\Gamma_{k_e}$, the minimum SE QoS $\mathcal{S}_{k_i}$, and the power budget for each AP.
Our objective is to optimally determine the AP mode selection $\qa$, power control coefficients, $\ETAI = \big\{\etamkI \big\}, \forall m \in \M, \forall k_i \in \Kii$ and $ \ETAE = \big\{\etamjE \big\}, \forall m \in \M, k_e \in \Kee$, and the PS matrices of $M$ SIM-attached APs $\boldsymbol{\Theta} \triangleq \{ \Phi^{m, \ell}, \forall m \in \M, \ell \in \mathcal{L} \}$. The optimization problem can be mathematically expressed as
\begin{subequations}\label{eq:ProblemFormulationorigin}
\begin{alignat}{2}
&\max_{\qa, \ETAI, \ETAE, \boldsymbol{\Theta} }      
&\qquad&\sum\nolimits_{k_e \in \Kee}{\Ex\big\{\mathcal{E}^{\mathrm{NL}}_{k_e}\big\}}~\label{obj:ObjectiveFunction1}\\
&\hspace{1.8em} \text{s.t.}
&        & \Ex\big\{\mathcal{E} ^{\mathrm{NL}}_{k_e}\big\} \geq \Gamma_{k_e}, \forall k_e \in \Kee,~\label{ct:HE_threshold}\\
&         &      &\text{SE}_{k_i} \geq \mathcal{S}_{k_i}, \forall k_i \in \Kii,~\label{ct:SINRThreshol}\\
&         &      &a_m\sum\nolimits_{k_i\in\Kii} \!\etamkI\! 
        \nonumber\\
&         &        &+\! (1\!-\!a_m) \sum\nolimits_{k_e\in\Kee}\!\etamjE\!\leq\! 1, \forall m \!\in\! \M,~\label{ct:etamkI}\\
&         &      &a_m \in \{0,1\},\forall m \in \M,~\label{ct:a_m}\\
&         &      &\theta^{\ell}_{s} \in [0, 2 \pi], \forall s \in \mathcal{S},~\label{ct:SIM:phaseshift}
\end{alignat}
\end{subequations}
where $ \Gamma_{k_e}$ and $\mathcal{T}_{k_i}$ are the QoS requirements for WPT and WIT operations, respectively, while \eqref{ct:etamkI} is the power constraint at each AP, obtained from~\eqref{eq:transmitpowerconstraint}. Problem~\eqref{eq:ProblemFormulationorigin} is an immensely complex problem due to its mixed-integer and non-convex nature. To solve this problem, we first present an autonomous DRL-based solution, followed by a robust two-stage optimization approach that leverages structured convex approximations.

\section{CTDE MADRL Strategy}~\label{sec:CTDE}
In this section, we propose a scalable and autonomous learning-based framework to solve the non-convex optimization problem defined in $(\mathcal{P}1)$. While DRL-based approaches have shown promise in solving mixed-integer problems without requiring manual convexification, the CTCE strategy faces a critical scalability bottleneck in our proposed system. In particular, the joint optimization of SIMs' phase shifts and power allocation creates a global action space with a dimension of $M \!\times\! (L \!\times\! S \!+\! K \!+\!1)$, which grows exponentially with the number of APs $(M)$. The high-dimensional action space complicates the learning process and impedes convergence. To address this, we decentralize the action space across the APs, thereby reducing the dimensionality for each individual actor to a more manageable local action space of size $(L\!\times\! S \!+\! K\!+\!1)$. This distributed learning approach removes the dimensionality bottleneck and facilitates effective convergence in large-scale networks. To further enhance learning efficiency and guarantee that the agents produce feasible outputs, we introduce two crucial mechanisms: a constraint-satisfied normalization process to handle the mixed-integer and bounded constraints $(30d)$ and $(30f)$, and max-min reward scaling to guide the agents to simultaneously maximize $(30)$ and satisfy $(30a),(30b)$. As a result, these techniques enable the CTDE framework to stabilize the learning process, then instantaneously determine optimal and feasible solutions in dynamic large-scale CF-mMIMO SWIPT systems.

\vspace{-0.8em}
\subsection{Markovian-Based Problem Transformation}~\label{sec:markovian}
We begin by transforming the system proposed in Section~\ref{sec:Sysmodel} into a Markovian environment tailored for agents, which are the $M$ resource-equipped APs. Let $t \in \mathcal{T}$ denote the time step of the learning interval. At each time step $t$, agent $m$ observes its local LSFC to all the receivers, which can be mathematically presented as
\begin{align}~\label{eq:MLstate}
    \localstatem[t] \triangleq \big[\Sigma^{\mathrm{NL}}[t\!-\!1], \BETA(m,:)[t] \big],
\end{align}
where $\Sigma^{\mathrm{NL}}[t\!-\!1] \triangleq \sum\nolimits_{k_e \in \Kee}{\Ex\big\{\mathcal{E}^{\mathrm{NL}}_{k_e}\big\}}$ is the sum-HE metric of the previous step and $\BETA(m,:)\triangleq \{ \barbetamk \}, \forall m\in \M$ is the LSFC observed by agent $m$.
Then, each agent $m$ determines its local joint action containing its operation mode $a_m$, power allocation $\ETAI(m,:)\!\triangleq\!\{ \etamkI, \forall k_i \in \Kii \}$, and $\ETAE(m,:)\!\triangleq\!\{ \etamjE, \forall k_e \in \Kee \}$, and $L \times S$ SIM's PS matrix $\boldsymbol{\Phi}^{m} \triangleq \{ \delta_{s}^{\ell}, \forall \ell \in \mathcal{L}, \forall s \in \mathcal{S} \}$.
However, the ranges of values of the optimization variables are distinct compared to each other. Thus, an execution normalization process should be carried out appropriately.

\begin{table*}[ht]
\centering
\caption{Complexity comparison across DRL training components}
\vspace{-0.5em}
\begin{tabular}{|c|c|c|c|c|}
\hline
\textbf{Training Components} & \textbf{Off-policy MADRL} & \textbf{Off-policy SADRL} & \textbf{On-policy MADRL} & \textbf{On-policy SADRL} \\
\hline
Total Neurons 
& $M \cdot P_{A,\mathrm{MA}} + P_{C,\mathrm{MA}}$ 
& $P_{A,\mathrm{SA}} + P_{C,\mathrm{SA}}$ 
& $M \cdot P_{A,\mathrm{MA}} + P_{C,\mathrm{MA}}$ 
& $P_{A,\mathrm{SA}} + P_{C,\mathrm{SA}}$ \\
\hline

Actor Forward Pass
& $\mathcal{O}\left(M \cdot L_{A,\mathrm{MA}} \cdot d_{lo} \cdot d_{la}\right)$ 
& $\mathcal{O}\left(L_{A,\mathrm{SA}} \cdot d_{go} \cdot d_{ga}\right)$ 
& $\mathcal{O}\left(M \cdot L_{A,\mathrm{MA}} \cdot d_{lo} \cdot d_{la}\right)$ 
& $\mathcal{O}\left(L_{A,\mathrm{SA}} \cdot d_{go} \cdot d_{ga}\right)$ \\
\hline

Critic Forward Pass 
& $\mathcal{O}\left(L_{C,\mathrm{MA}} \cdot (d_{go} + d_{ga})^2\right)$ 
& $\mathcal{O}\left(L_{C,\mathrm{SA}} \cdot (d_{go} + d_{ga})^2\right)$ 
& $\mathcal{O}\left(L_{C,\mathrm{MA}} \cdot (d_{go} + d_{ga})^2\right)$ 
& $\mathcal{O}\left(L_{C,\mathrm{SA}} \cdot (d_{go} + d_{ga})^2\right)$ \\
\hline

Critic Propagation  
& $\mathcal{O}\left(P_{C,\mathrm{MA}}\right)$ 
& $\mathcal{O}\left(P_{C,\mathrm{SA}}\right)$ 
& $\mathcal{O}\left(P_{C,\mathrm{MA}}\right)$ 
& $\mathcal{O}\left(P_{C,\mathrm{SA}}\right)$ \\
\hline

Actor Propagation  
& $\mathcal{O}\left(M \cdot P_{A,\mathrm{MA}}\right)$ 
& $\mathcal{O}\left(P_{A,\mathrm{SA}}\right)$ 
& $\mathcal{O}\left(M \cdot P_{A,\mathrm{MA}}\right)$ 
& $\mathcal{O}\left(P_{A,\mathrm{SA}}\right)$ \\
\hline

Replay Buffer Sampling 
& $\mathcal{O}\left(B_e \cdot (d_{go} + d_{ga})\right)$ 
& $\mathcal{O}\left(B_e \cdot (d_{go} + d_{ga})\right)$ 
& 0 
& 0 \\
\hline

Target Network Update  
& $\mathcal{O}\left(P_{C,\mathrm{MA}} + M \cdot P_{A,\mathrm{MA}}\right)$ 
& $\mathcal{O}\left(P_{C,\mathrm{SA}} + P_{A,\mathrm{SA}}\right)$ 
& 0 
& 0 \\
\hline

Advantage Estimation  
& 0 
& 0 
& $\mathcal{O}\left(T \cdot d_{go}\right)$ 
& $\mathcal{O}\left(T \cdot d_{go}\right)$ \\
\hline

PPO Clipping + Ratio 
& 0 
& 0 
& $\mathcal{O}\left(T \cdot d_{la}\right)$ 
& $\mathcal{O}\left(T \cdot d_{ga}\right)$ \\
\hline
\makecell[c]{Optimal decision-making\\post training} 
& $\mathcal{O}\left( L_{A,\mathrm{MA}} \cdot d_{lo} \cdot d_{la} \right)$
& $\mathcal{O}\left( L_{A,\mathrm{SA}} \cdot d_{go} \cdot d_{ga} \right)$ 
& $\mathcal{O}\left( L_{A,\mathrm{MA}} \cdot d_{lo} \cdot d_{la} \right)$
& $\mathcal{O}\left( L_{A,\mathrm{SA}} \cdot d_{go} \cdot d_{ga} \right)$ 
\\\hline
\end{tabular}
\begin{tablenotes}\footnotesize
\item Abbreviations: actor (A), batch size (\text{$B_e$}), critic (C), global/local action (ga/la), global/local observation (go/lo), multi/single-agent (MA/SA)
\end{tablenotes}
\vspace{-1em}
\label{tab:complexity_comparison}
\end{table*}

\vspace{-0.5em}
\subsection{Constraint-Satisfied Normalization Process}~\label{sec:normalization}
We first introduce ${\rawlocalactionm}[t] \in \R^{1 + K + L \times S} \in [0,1]$ output after activation process from agent $m$, given by
\begin{align}~\label{eq:LocalAction}
    \rawlocalactionm[t] \triangleq \vecOp[\breve{a}_m, \boldsymbol{\varrho}_{m}^{\mathtt{I}}, \boldsymbol{\varrho}_{m}^{\mathtt{E}}, \boldsymbol{\Delta}_{m} ],
\end{align}
where $\boldsymbol{\varrho}_{m}^{\mathtt{I}} \triangleq \{ \varrhomkI, \forall k_i \in \Kii \}$, $\boldsymbol{\varrho}_{m}^{\mathtt{E}} \triangleq \{ \varrhomjE, \forall k_e \in \Kee \}$,
$\boldsymbol{\Delta}_{m} \triangleq \{ \delta_{s}^{\ell}, \forall \ell \in \mathcal{L}, \forall s \in \mathcal{S}  \}$. We first apply a step function to discretize the continuous variable $\breve{a}_m$ into a binary decision, where $a_m = 1$ if $\breve{a}_m \geq 0.5$.
Then, we implement a normalization step $\theta^{\ell}_{s} \triangleq 2 \pi \delta_{s}^{\ell}$, thus the constraint~\eqref{ct:SIM:phaseshift} is satisfied. In addition, in order to satisfy the constraint \eqref{ct:etamkI}, we implement $\textit{softmax}$ normalization for each power allocation of AP $m$. Mathematically speaking, 
\begin{align}
    \etamkI &\!=\! \exp( \varrhomkI) \exp\big( \sum\nolimits_{k'_i \in K_\mathtt{I}}{\varrhomkpI} \big)^{-1}, \text{ and},
    \nonumber\\
    \etamjE &\!=\! \exp( \varrhomjE) \exp\!\big(\! \sum\nolimits_{k'_e \in K_\mathtt{E}}\!{\varrhomjpE} \big)^{\!-1}, \forall m \!\in\! \M.
\end{align} 
Eventually, the \textit{constraint-satisfied} action of the agent $m$ is mathematically expressed as
\begin{align}
    \localactionm[t] \triangleq \vecOp\big[ a_{m}, \ETAI(m,:), \ETAE(m,:), \boldsymbol{\Phi}^{m}  \big].
\end{align}
\subsection{Normalized Joint Reward Function}
\vspace{-0.2em}
Building upon the closed-form expressions derived in Proposition~\ref{Theorem:SE:PPZF} and~\ref{Theorem:RF:MRT}, we observe a direct and constructive influence of $\FmRIS, \forall m \in \M$ on the overall SWIPT performance. Particularly, the Frobenius norm of the equivalent SIM channel, $\Vert\FmRIS\Vert_{\mathrm{F}}$ plays a critical role in enhancing both signal focusing and energy delivery. Thus, we incorporate the Frobenius norm of the SIM-assisted EM wave propagation, $\Vert\FmRIS\Vert_{\mathrm{F}}$, into the reward structure alongside with the maximum objective and QoS constraints in~\eqref{eq:ProblemFormulationorigin}. Mathematically speaking
\begin{align}~\label{eq:MLreward}
    r[t] \triangleq \lambda_{r} \tilde{\Delta}_{\mathrm{HE}}[t] + (1-\lambda_{r}) \tilde{\mathsf{Q}}[t] - \lambda_{\mathrm{SE}},
\end{align}
where $\Delta_{\mathrm{HE}}[t] \triangleq \Sigma^{\mathrm{NL}}[t] -  \Sigma^{\mathrm{NL}}[t\!-\!1]$, 
$\mathsf{Q}[t] \triangleq \sum\nolimits_{m \in \M}{\Vert \FmRIS \Vert_{\mathrm{F}}}$,
$\lambda_{r}\in[0,1]$ is the trade-off parameter to balancing the WPT improvement against SIMs' EM wave beamforming, $\lambda_{\mathrm{SE}}$ is a penalty for violating~\eqref{ct:SINRThreshol}, $\tilde{\Delta}_{\mathrm{HE}}[t]$ and $\tilde{\mathsf{Q}}[t]$ are the normalized values using dynamic min-max tracking. To handle scaling disparity between two distinct values, we update the normalization bounds at each $t$ using:
\begin{align}~\label{eq:maxmin_scaling}
    \tilde{\Delta}_{\mathrm{HE}}[t] \!=\! \frac{\Delta_{\mathrm{HE}}[t] \!-\! \min{\Delta_{\mathrm{HE}}}}{\max{\Delta_{\mathrm{HE}}} \!-\! \min{\Delta_{\mathrm{HE}}}}
    ,
    \tilde{\mathsf{Q}}[t] \!\!=\!\! \frac{\mathsf{Q}[t] - \min{\mathsf{Q}}}{\max{\mathsf{Q}} - \min{\mathsf{Q}}}.
\end{align}
To prevent bias values and changing scales, the min-max thresholds are tracked from the previous step, i.e., $\min{x}[t]\triangleq\min(\min{x}[t\!-\!1], \min{x})$ and $\max{x}[t]\triangleq\min(\max{x}[t\!-\!1], \max{x})$,
where $x\in\{\Delta_{\mathrm{HE}},\mathsf{Q} \}$. To this end, the normalization joint reward function serves a dual purpose: it stabilizes the learning pattern for improved centralized learning convergence and encourages each agent to adaptively execute $\boldsymbol{\Theta}$ to enhance SWIPT metrics over the learning time.
\begin{Remark}
    While any RL environment must satisfy the Markovian property, an inefficient designed environment may not be equally effective for both off-policy and on-policy trainings, as numerically shown in~\cite[Section IV.C]{cite:thien_iotj_2023}. 
    However, as multiple variables in the reward~\eqref{eq:MLreward} are normalized within predefined ranges, the resulting reward signal becomes more stable across episodes, contributing to reduced variance in policy gradient updates and improved convergence behavior. Consequently, our proposed Markovian environment effectively supports long-term learning and stability within any centralized training with the CTDE framework. In this framework, each of the $M$ agents independently observes its local state~\eqref{eq:MLstate} and selects its action~\eqref{eq:LocalAction} in a decentralized manner. During training, a centralized critic leverages global information, including the joint state and joint actions of all agents. This design enables agents to cooperatively learn policies that maximize a global reward, while respecting constraints in~\eqref{eq:ProblemFormulationorigin}.
\end{Remark}
 
\subsection{Complexity Analysis}
In this {subsection}, we analyze and compare the model complexity of on-policy and off-policy \textbf{MADRL} and single-agent \textbf{SADRL}. For off-policy methods, the replay buffer stores experiences $E \triangleq [\qs[t], \qa[t], r[t], \qs[t+1]]$ with a buffer length $B$ and batch size $B_e$. We denote $d_{\mathrm{lo}}\triangleq 1+K$, $d_{\mathrm{go}}\triangleq 1+MK$, $d_{\mathrm{la}}\triangleq 1+K+LS$, $d_{\mathrm{ga}}\triangleq M d_{\mathrm{la}}$, the dimensions of the global observation, local observation, global action, and local action, respectively. Utilizing fully-connected DNNs, we denote by $P_{A, \cdot}$ and $P_{C, \cdot}$ the number of neurons in the actor and critic networks, respectively. Table~\ref{tab:complexity_comparison} provides a comprehensive complexity comparison across all DRL frameworks.

\section{Two-Phrase Robust Joint AP Mode Selection and Power Control Design}~\label{sec:jap-pa}
In this section, we introduce a robust optimization scheme based on the convex optimization method to serve as a comparison benchmark for evaluating our proposed CTDE autonomous learning framework. We selected the combined heuristic and SCA approach for its tractability, strong theoretical convergence guarantees, and seamless integration with standard convex optimization solvers like CVX and MOSEK. However, applying them necessitates extensive convex reformulations and approximations, which results in high computational complexity, especially as the network scales. Therefore, a direct comparison is essential to evaluate the performance-complexity trade-off, assessing whether the significant reduction in computational overhead offered by the MADRL-based CTDE framework justifies any potential performance gap against the convex-based solution. We provide a step-by-step procedure to convexify the constraints in~\eqref{eq:ProblemFormulationorigin}, and a comprehensive numerical comparison between the convex-based solution and the DRL-based approach is presented in Section~\ref{sec:vb}.

\subsection{Heuristic Phase Shift Design (\textbf{HPS})}~\label{subsec:heu}
From the closed-form in Proposition~2, we notice that the SIM's PS directly affects the trace value of the SIM-affected channel with its Hermitian, which serves as a dominant term in the expression. Thus, we heuristically optimize the PSs by solving the problem:
\vspace{-0.5em}
\begin{equation}~\label{eq:heu:obj}
    \max_{\boldsymbol{\Theta}^{m}}{f(\boldsymbol{\Theta}) = \trace(\FmRIS \FmRIS^{\dag})}.
\end{equation}

\textbf{Algorithm}~\ref{alg:layerbylayerSIM} provides the layer-by-layer heuristic scheme while the complexity of the heuristic scheme is $\mathcal{O}(MCSL^{2} )$. Notably, the effectiveness of the heuristic search is proportional to the heuristic iteration $C$, as a higher $C$ allows for extensive explorations of the solution space. Thus, $C$ serves as a trade-off parameter between computational efficiency and optimization performance.

\vspace{-0.5em}
\subsection{Joint AP-mode Selection and Power Allocation (\textbf{JAPPA})}~\label{subsec:ppa}
Once the SIM's PS are optimally determined, the joint AP-mode selection and power allocation problem is then solved using a SCA method.
For subsequent convexify steps, we first reformulate the constraint~\eqref{ct:HE_threshold}. To expand the logistic function $\Lambda\big(\mathcal{Q}_{k_e}\big)$, we introduce $\mathcal{Q}_{k_e} \geq \Xi(\tilde{\Gamma}_{k_e})$, where $\tilde{\Gamma}_{k_e} = (1-\Omega) \Gamma_{k_e} + \phi\Omega$, where $\Xi(\tilde{\Gamma}_{k_e}) = \chi - \frac{1}{\xi}\Big( \frac{\phi -\tilde{\Gamma}_{k_e}}{\tilde{\Gamma}_{k_e}} \Big)$,
which is the inverse function of the logistic function. Then, we define $\boldsymbol{\epsilon} =\{\epsilon_{k_e}\geq 0, k_e\in\Kee\}$ where $\epsilon_{k_e}$ is an auxiliary variable for $\tilde{\Gamma}_{k_e}$, we obtain: $\tilde{\epsilon}_{k_e} \triangleq (1 - \Omega) \epsilon_{k_e} + \phi \Omega$. Subsequently, we transform the binary constraint~\eqref{ct:a_m} into continuous counterpart~\cite{Mohamed_projection_2024}
\vspace{-0.5em}
\begin{align}~\label{eq:am:continuous}
    a_m \in \{0,1\} \equiv a_m - a_m^2 \geq 0,  a_m \in [0,1].
\end{align}
To accelerate the convergence speed, we relax \eqref{eq:am:continuous} into the objective, and then transform the non-convex $a_m^2$. Herein, we use the concave lower bound, $ x^2 \geq x_0(2x - x_0)$, in which we respectively replace $a_m^2, a_m^{(n)}, a_m$ by $x^2, x_0, x$, and superscript ($n$) denotes the value of the involving variable produced after $(n - 1)$ iterations ($n \geq 1$)~\cite{Mohammadi:JSAC:2023}. Given fixed $\boldsymbol{\Theta}$ configuration after the first heuristic phrase, we reformulate~\eqref{eq:ProblemFormulationorigin} as
\vspace{-0.5em}
\begin{subequations}\label{eq:ProblemFormulationP2}
\begin{alignat}{2}
&\max_{\qa, \ETAI, \ETAE, \qe}      
&\hspace{0.5em}&\sum\nolimits_{k_e\in\Kee}{\epsilon_{k_e}} \nonumber\\
&         &      &- \lambda^{\mathrm{pen}}\sum\nolimits_{m\in\M}{a_m - a_{m}^{(n)} \big(2a_m - a_{m}^{(n)}\big)}\label{obj:ObjectiveFunction2}\\
&\hspace{1.5em}\text{s.t.} 
&         & Q_{k_e} \geq \Xi(\tilde{e}_{k_e}),~\forall k_e \in \Kee,~\label{ct:auxilaryvariable2a}\\
&         &      &\text{SE}_{k_i} \geq \mathcal{S}_{k_i}, \forall k_i \in \Kii,~\label{ct:SINRThreshol31}\\
&         &      &\epsilon_{k_e} \geq \Gamma_{k_e}, \forall k_e \in \Kee,~\label{ct:auxilaryvariable2b2}\\
&         &      &\sum\nolimits_{k_i\in\Kii}\!\!{\etamkI}\! \leq a_{m}^{(n)} \!\big(2a_m \!-\! a_{m}^{(n)}\big), \forall m \in \M,~\label{ct:etamkI22}\\
&         &      &\sum\nolimits_{k_e\in\Kee}{\etamjE} \leq 1 - a_{m}^{2}, \forall m \in \M,~\label{ct:etamjE22}\\
&         &      &0 \leq a_m \leq 1,\forall m \in \M.~\label{ct:a_m2relax2}
\end{alignat}
\end{subequations}

\begin{algorithm}[t]
\caption{Layer-by-layer heuristic search for $L$-layer SIM}~\label{alg:layerbylayerSIM}
    \begin{algorithmic}[1]
    \small
        \STATE Initiate SIM with stochastic ${\Phi}^{m, \ell}$ in range of $[0, 2\pi)$, $\forall \ell \in \mathcal{L}, \forall m \in \M$
        \STATE \textbf{for} $m = 1$ \textbf{to} $M$:
        \STATE \quad \textbf{for} $\ell = 1$ \textbf{to} $L$:
        \STATE \quad \quad \textbf{for} $\mathtt{it} = 1$ \textbf{to} $C$:
            \STATE  \quad \quad \quad 
            Generate $\boldsymbol{\Phi}^{m \ell}$ and compute $f(\mathtt{it})$ according to~\eqref{eq:heu:obj}
            \STATE  \quad \quad \quad 
            Store $\boldsymbol{\Phi}^{m \ell}(\mathtt{it}), \mathrm{obj}(\mathtt{it})$ and its index
            \STATE \quad \quad 
            Search the maximum $\mathrm{obj}(\mathtt{it}^*)$ value and its index $\mathtt{it}^*$
            \STATE \quad \quad
            Assign $\mathtt{it}^*$-th $\boldsymbol{\Phi}^{m \ell}$ as optimal $(\boldsymbol{\Phi}^{m \ell})^{*}$
            \STATE \quad \quad Move to $\ell + 1$ layer until the last layer
            \STATE \quad  Move to $m + 1$ AP until the last AP
        \RETURN $(\boldsymbol{\Theta}^{m})^{*} \triangleq \{ (\boldsymbol{\Phi}^{m \ell})^{*}, \forall m \in \M, \ell \in \mathcal{L} \}$.
    \end{algorithmic}
\end{algorithm}
\setlength{\textfloatsep}{0.3cm}

We first notice that the non-convex nature of~\eqref{ct:auxilaryvariable2a} is due to the product of the two variables, i.e. $a_m \etamkI$ and $a_m \etamjE$, and the inverse function $\Xi(\tilde{\Gamma}_{k_e})$. To address the non-convex terms, we use the upper bound of $\Xi(\tilde{\epsilon}_{k_e})$ as
\vspace{-0.5em}
\begin{align}~\label{eq:CvxUBINV}
    \Xi(\tilde{\epsilon}_{k_e}) 
    &\!\leq\!
    \chi 
    \!-\! 
    \frac{1}{\xi} \bigg(\! \ln\Big(\frac{\phi - \tilde{\epsilon}_{k_e}}{\tilde{\epsilon}_{k_e}^{(n)}} \Big)
    \!-\! 
    \frac{\tilde{\epsilon}_{k_e} - \tilde{\epsilon}_{k_e}^{(n)}}{\tilde{\epsilon}_{k_e}^{(n)}} \bigg)\triangleq \tilde{\Xi}(\tilde{\epsilon}_{k_e}),
\end{align}
where $\tilde{\epsilon}_{k_e}^{(n)}$ is the $(n-1)$-th iteration of $\tilde{\epsilon}_{k_e}$
To this end, we recast the constraint~\eqref{ct:auxilaryvariable2a} as
\vspace{0em}
\begin{align}~\label{eq:44}
    &(\tau_c - \tau) \Snn \snrdl 
    \sum\nolimits_{m \in \M} 
    \!\Big(
    \etamjE \!\bar{d}_{m k_e} 
    \!+\!
    \sum\nolimits_{k'_e \in \mathcal{P}_k \setminus k_e}
    {\etamjpE \bar{e}_{m, k_e k'_e} }
    \nonumber\\
    &\hspace{-0.3em}+\!
    \sum\nolimits_{k'_e \notin \mathcal{P}_k \setminus k_e}
    {\!\etamjpE \bar{f}_{m, k_e k'_e} }
    \Big)
    \!+\!
    (\tau_c - \tau) \Snn \snrdl 
    \sum\nolimits_{m \in \M} \! a_m \mathcal{Z}_{m k_e} 
    \nonumber\\
    &\hspace{-0.3em}+
    (\tau_c - \tau) \Snn
    \geq
    \tilde{\Xi}(\tilde{\epsilon}_{k_e}),
\end{align}
where
\begin{align}
    &\hspace{-0.3em}\mathcal{Z}_{m k_e} = 
    \sum\nolimits_{k_i \in \Kii }
    \etamkI (1/N) \bar{h}_{m, k_e}
    - \etamjE \bar{d}_{m, k_e} 
    \nonumber\\
    &\hspace{-0.5em}- 
    \sum\nolimits_{k'_e \in \mathcal{P}_k \setminus k_e}
    {\etamjpE \bar{e}_{m, k_e k'_e} }
    -
    \sum\nolimits_{k'_e \notin \mathcal{P}_k \setminus k_e}
    {\etamjpE \bar{f}_{m, k_e k'_e} }
    .
\end{align}
Using~$4xy = (x+y)^2 - (x-y)^2$, we further recast~\eqref{eq:44} as
\begin{align}~\label{eq:46}
    &\hspace{-0.3em}4(\tau_c \!-\! \tau) \Snn \snrdl 
    \!\sum\nolimits_{m \in \M} 
    \!\!\Big(
    \etamjE \bar{d}_{m k_e} 
    \!\!+\!\!
    \sum\nolimits_{k'_e \in \mathcal{P}_k \setminus k_e}
    {\etamjpE \bar{e}_{m, k_e k'_e} }
    \nonumber\\
    &\hspace{-0.3em}
    \!\!+\!\!
    \sum\nolimits_{k'_e \notin \mathcal{P}_k \setminus k_e}
    \!\!\!{\!\!\etamjpE \bar{f}_{m, k_e k'_e} }
    \Big)
    \!\!+\!\! 
    (\tau_c \!-\! \tau) \Snn \snrdl
    \!\!\sum\nolimits_{m \in \M} (a_m \!+\!\! \mathcal{Z}_{m k_e})^{2}
    \nonumber\\
    &\hspace{-0.3em}\!\!+\!\!
    4(\tau_c \!-\! \tau) \Snn
    \geq 
    4 \tilde{\Xi}(\tilde{\epsilon}_{k_e})
    +
    (\tau_c \!-\! \tau) \Snn \snrdl
    \!\!\sum\nolimits_{m \in \M} \!(a_m \!-\!\! \mathcal{Z}_{m k_e})^{2}.
\end{align}
The above inequality is still non-convex, due to the presence of $(a_m \!+\!\! \mathcal{Z}_{m k_e})^{2}$. Therefore, using SCA, we get
\begin{align}~\label{ct:auxilaryvariable2a_convex}
    &4(\tau_c - \tau) \Snn \snrdl 
    \!\sum\nolimits_{m \in \M} 
    \!\Big(\!
    \etamjE \bar{d}_{m k_e} 
    \!\!+\!\!
    \sum\nolimits_{k'_e \in \mathcal{P}_k \setminus k_e}
    \!\!{\etamjpE \bar{e}_{m, k_e k'_e} }
    \nonumber\\
    &+
    \sum\nolimits_{k'_e \notin \mathcal{P}_k \setminus k_e}
    {\etamjpE \bar{f}_{m, k_e k'_e} }
    \Big)
    +
    4(\tau_c - \tau) \Snn
    + 
    (\tau_c - \tau) \Snn \snrdl
    \nonumber\\
    &\times
    \sum\nolimits_{m\in\M}
    \big(a_m^{(n)} + \mathcal{Z}_{m k_e}^{(n)} \big) \Big(2\big(a_m + \mathcal{Z}_{m k_e} \big) - a_m^{(n)} - \mathcal{Z}_{m k_e}^{(n)} \Big)
    \nonumber\\
    &\geq 
    4 \tilde{\Xi}(\tilde{\epsilon}_{k_e})
    +
    (\tau_c - \tau) \Snn \snrdl
    \sum\nolimits_{m \in \M} (a_m - \mathcal{Z}_{m k_e})^{2},
\end{align}
where we used $x^2 \geq x_0(2x - x_0)$ and replaced $x$ and $x_0$ by $a_m + \mathcal{Z}_{mj}$ and $a_m^{(n)} + \mathcal{Z}_{mj}^{(n)}$, respectively. 

To this end, \eqref{ct:SINRThreshol31} is equivalent to
\begin{align}~\label{eq:ctSINR1}
    \frac{\snrdl q_{k_{i}}^2}{\mathcal{T}_{k_i}}
    \!&\geq\! 
    \snrdl
    \sum\nolimits_{k'_i \in \mathcal{P}_k \setminus k_i}
    \!\!
    q_{k'_i}^2
    \!+\!
    \snrdl 
    \!\!
    \sum\nolimits_{m\in\M}
    \!
    \sum\nolimits_{k'_i \in \Kii}
    \!\!
    a_m \etamkpI \bar{e}_{m k_i\nonumber}\\
    &+ 
    \snrdl
    \sum\nolimits_{m\in\M}
    \sum\nolimits_{k_e \in \Kee}
    \etamjE \bar{e}_{m k_i}
    \nonumber\\
    &-
    \snrdl
    \sum\nolimits_{m\in\M}
    \sum\nolimits_{k_e \in \Kee}
    a_m \etamjE \bar{e}_{m k_i} + 1
    \nonumber
\\
    \!&\geq\! 
    \snrdl
    \sum\nolimits_{k'_i \in \mathcal{P}_k \setminus k_i}
     \!\!
    q_{k'_i}^2
    \!\!+\!\! 
    \snrdl 
    \!\!
    \sum\nolimits_{m\in\M}
    \!
    \sum\nolimits_{k'_i \in \Kii}
    \!\!
    a_m \etamkpI \bar{e}_{m k_i}
    \nonumber\\
    &+
    \snrdl
    \sum\nolimits_{m\in\M}
    \sum\nolimits_{k_e \in \Kee} 
    a_m \omega_m \bar{e}_{m k_i} + 1,
\end{align}
where we have replaced $q_{k_i} \triangleq  \sum\nolimits_{m\in\M} \alphaPZFmki \sqrt{a_m \etamkI}$, $\forall k_i \in \Kii$, $\omega_m \triangleq \etamI-\etamE$, $\etamI \triangleq \sum\nolimits_{k'_i\in \Kii} \etamkpI$, and $\etamE \triangleq \sum\nolimits_{k_e\in \Kee }\etamjE$. To convexify $q_{k_i}$, we define its ($n-1$)-th iteration, i.e., $q_{k_i}^{(n)}\triangleq \sum\nolimits_{m\in\M} \alphaPZFmki \sqrt{ a_m^{(n)}\etamkIn }$, then apply $x^2 \geq x_0(2x-x_0)$ for $x\!\triangleq\!q_{k_i}$ and $x_0\!\triangleq\!q_{k_i}^{(n)}$ to obtain
\vspace{-0.2em}
\begin{align}
    &\frac{\snrdl q_{k_i}^{(n)} }{\mathcal{T}_{k_i}} \Big(
    2q_{k_i} - q_{k_i}^{(n)}
    \Big) 
    \geq
    1
    +
    \snrdl  
    \sum\nolimits_{k'_i \in \mathcal{P}_k \setminus k_i}
    q_{k'_{i}}^{2}
    \nonumber\\
    &+
    \snrdl \sum\nolimits_{m\in\M}
    a_m \omega_m \bar{e}_{m k_i} 
    +
    \snrdl
    \sum\nolimits_{m\in\M}
    \etamE \bar{e}_{m k_i}.
\end{align}
The above inequality is still non-convex due to $a_m \omega_m$ in the right-hand side, thus we use~$4a_m \omega_m\!=\!(a_m\!+\!\omega_m)^2\!-\!(a_m\!-\!\omega_m)^2$ and then utilize  $x^2 \geq x_0(2x - x_0)$, for $x\!\triangleq\!a_m\!-\!\omega_m$ and $x_0\!\triangleq\!a_m^{(n)}\!-\!\omega_m^{(n)}$ to get the following convex constraint
\vspace{-0.2em}
\begin{align}~\label{eq:ctSINR3}
    &\frac{4 \snrdl q_{k_i}^{(n)} }{\mathcal{T}_{k_i}} \Big(
    2q_{k_i} \!-\! q_{k_i}^{(n)}
    \Big)
    \! + \!
    \rho_{d}\!\sum\nolimits_{m\in\M} \!\! 
    \bar{e}_{m k_i}
    \Big(\! a_{m}^{(n)} \!-\!\omega_{m}^{(n)}\Big)
    \nonumber\\
    &\times
    \Big(\!2\big(a_{m} \!-\! \omega_{m} \big) 
    \!\!-\!\! 
    \big(\!a_{m}^{(n)} \!-\! \omega_{m}^{(n)} \!\big)
    \!\! \Big)
    \geq
    4
    +
    4\snrdl
    \!
    \sum\nolimits_{k'_i \in \mathcal{P}_k \setminus k_i}{q_{k'_i}}
    \nonumber\\
    &+ 
    \snrdl
    \!
    \sum\nolimits_{m\in\M} \bar{e}_{m k_i} (a_m \!+\! \omega_m)^{2} 
    \!+\!
    4\snrdl
    \!
    \sum\nolimits_{m\in\M}
    \etamE \bar{e}_{m k_i}.
    \end{align}
The optimization problem~\eqref{eq:ProblemFormulationP2}  is now recast as
\vspace{-0.2em}
\begin{subequations}\label{opt:JAP:final}
\begin{alignat}{2}
&\max_{\qa, \ETAI, \ETAE, \qe}      
&\qquad&\sum\nolimits_{k_e\in\Kee}{\epsilon_{k_e}} \nonumber\\
&         &      &- \lambda^{\mathrm{pen}}\sum\nolimits_{m\in\M}{a_m - a_{m}^{(n)} \big(2a_m - a_{m}^{(n)}\big)}\label{opt:JAP:final:obj}\\
&\hspace{0.5em}\text{s.t.} 
&         &~\eqref{ct:auxilaryvariable2a_convex},~\forall k_e \in \Kee,~\label{opt:JAP:final:ct1}\\
&         &      &~\eqref{eq:ctSINR3},~\forall k_i \in \Kii,
               ~\label{opt:JAP:final:ct2}\\
&         &      &~\eqref{ct:auxilaryvariable2b2}-\eqref{ct:a_m2relax2}.~\label{opt:JAP:final:ct3}
\end{alignat}
\end{subequations}
Problem~\eqref{opt:JAP:final} is convex, and thus it can be solved using CVX. In \textbf{Algorithm~\ref{alg1}}, we outline the main steps to solve problem~\eqref{opt:JAP:final}, where $\widetilde{\qx} \triangleq \{\qa, \ETAI, \ETAE, \qe\}$ and $\widehat{\mathcal{F}} \triangleq\{\eqref{opt:JAP:final:ct1},~\eqref{opt:JAP:final:ct2},~\eqref{opt:JAP:final:ct3}\}$ are convex feasible sets. Starting from a random point $\widetilde{\qx}\in\widehat{\mathcal{F}}$, we solve \eqref{opt:JAP:final} to obtain its optimal solution $\widetilde{\qx}^*$, and use $\widetilde{\qx}^*$ as an initial point in the next iteration. The proof of this convergence property follows similar steps as the proof in \cite[Proposition 2]{vu18TCOM}, and hence, is omitted.

\textbf{Complexity analysis:} The relaxed optimization problem~\eqref{opt:JAP:final} each requires a computational complexity of $\mathcal{O}(\sqrt{A_l + A_q}(A_v + A_l + A_q)(A_v)^2)$, per iteration~\cite{Mohammadi:JSAC:2023}, where $A_v$ represents the number of real-valued scalar variables, while $A_l$ and $A_q$ denote the number of linear and quadratic constraints, respectively. For ease of presentation, if we let $K_E=K_I$, we have $A_v=M(K_e+K_i+1)+K_e$, $A_l=2M+K_e$, and $A_q=M+K_e+K_i$.
\begin{algorithm}[t]
\caption{Joint AP mode operation and power allocation}
\begin{algorithmic}[1]
\label{alg1}
\STATE \textbf{Initialize}: $n\!=\!0$, 
$\lambda^{\mathrm{pen}} > 1$, random initial point $\widetilde{\qx}^{(0)}\!\in\!\widehat{\mathcal{F}}$.
\REPEAT
\STATE Update $n=n+1$
\STATE Solve \eqref{opt:JAP:final} to obtain its optimal solution $\widetilde{\qx}^*$
\STATE Update $\widetilde{\qx}^{(n)}=\widetilde{\qx}^*$
\UNTIL{convergence}
\end{algorithmic}
\end{algorithm}
\setlength{\textfloatsep}{0.1cm}

\begin{figure*}[t]
\vspace{-0.5cm}
    \centering
    \begin{minipage}[t]{0.32\textwidth}
        \centering
         \includegraphics[trim=0 0cm 0cm 0cm,clip,width=1.12\textwidth]{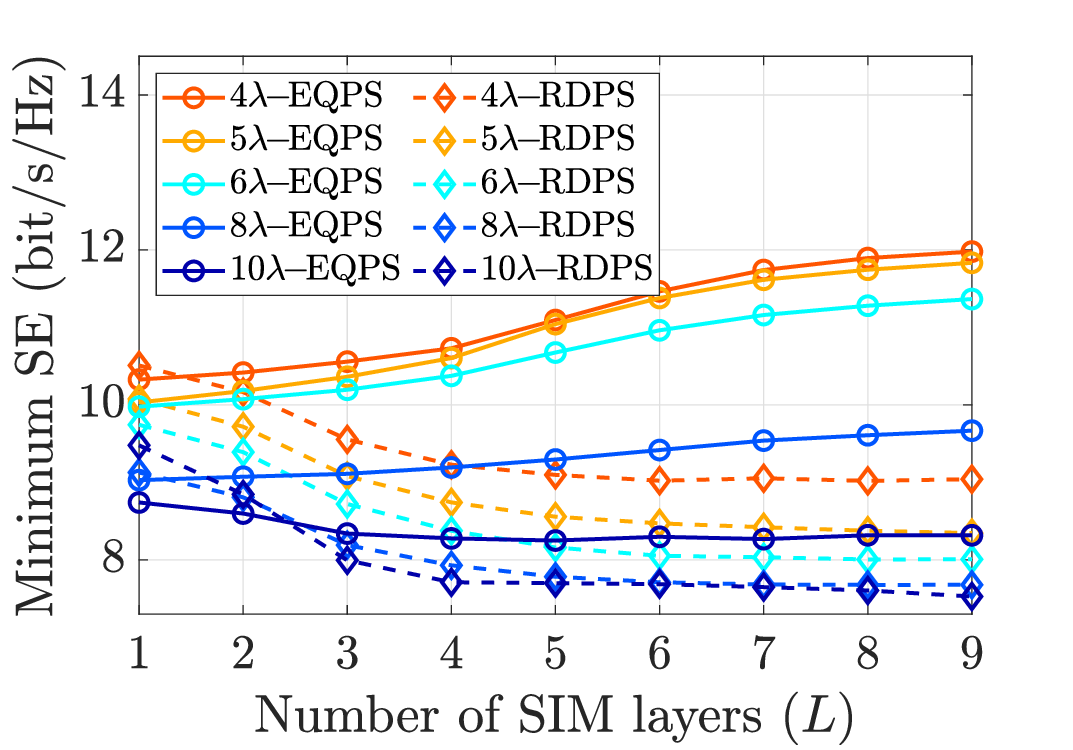}
         \vspace{-1.5em}
        \caption{\small Minimum SE versus $T_{\mathrm{SIM}}$ ($S=36$, $M = 10, N = 20, \mathrm{PRF}^{\mathtt{I}} = 0, \mathrm{PRF}^{\mathtt{I}} = 3$).}
        \label{fig:SE_v_thickness}
    \end{minipage}
    \hfill
    \begin{minipage}[t]{0.32\textwidth}
        \centering
          \includegraphics[trim=0 0cm 0cm 0cm,clip,width=1.12\textwidth]{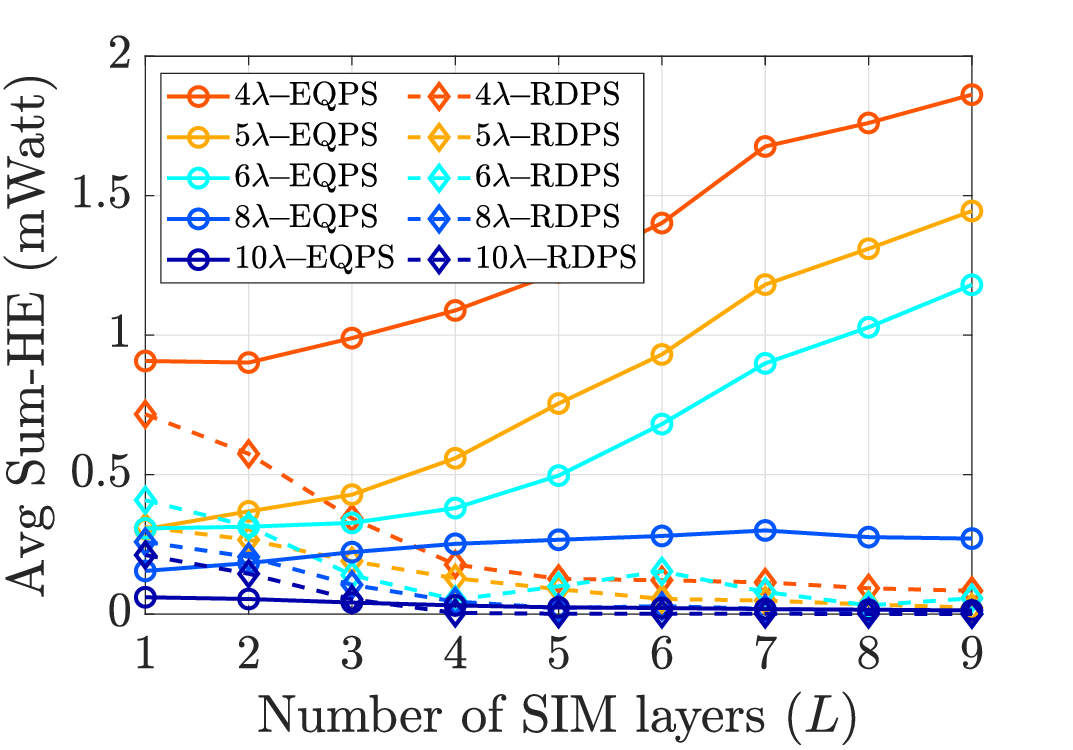}
          \vspace{-1.5em}
        \caption{\small Sum-HE versus $T_{\mathrm{SIM}}$ ($S=36$, $M = 10, N = 20, \mathrm{PRF}^{\mathtt{I}} = 0, \mathrm{PRF}^{\mathtt{I}} = 3$).}
        \label{fig:HE_v_thickness}
    \end{minipage}
    \hfill
    \begin{minipage}[t]{0.32\textwidth}
        \centering
         \includegraphics[trim=0 0cm 0cm 0cm,clip,width=1.12\textwidth]{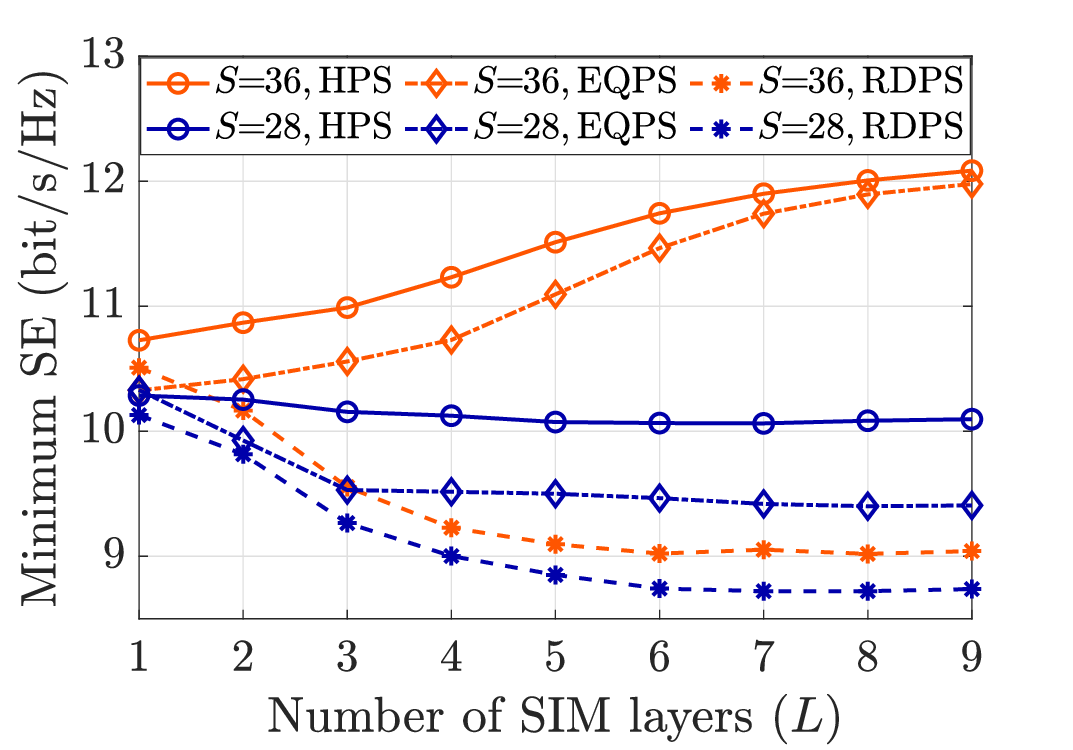}
         \vspace{-1.5em}
        \caption{\small Minimum SE versus $S$ ($T_{\mathrm{SIM}}\!=\!4\lambda$, $M\!=\!10, \mathrm{PRF}^{\mathtt{I}} = 0, \mathrm{PRF}^{\mathtt{I}} = 3$).}
        \label{fig:figurec}
    \end{minipage}
\vspace{-0.5em}
\end{figure*}

\section{Numerical Results}~\label{sec:PerformanceAnalysis}
We present numerical results to evaluate the performance of the SIM-assisted CF-mMIMO SWIPT system. We consider that the APs, $K_i=3$ IRs and $K_e=4$ ERs, are randomly distributed in a square of $0.1 \times 0.1$ km${}^2$. The heights of the SIM-attached APs and the receivers are $15$~m, and 1.65~m, respectively. We set $\tau_c = 200$, $\Snn = -92$ dBm, $\rho_d = 1$~W and $\rho_u = 0.2$ W.  In addition, we set the NLEH model's parameters as $\xi = 150$, $\chi = 0.024$, and $ = 0.024$~W~\cite{Mohammadi:TC:2024, Hua:WCNC:2024}. The LSFCs are generated following the three-slope propagation model from \cite{cite:HienNgo:cf01:2017}. Shadow fading follows a log-normal distribution with a standard deviation of $8$ dB. We denote the number of symbols for the UL channel estimation phase as $\tau = K_{\mathtt{I}} + K_{\mathtt{I}} - \mathrm{PRF}$, where $\mathrm{PRF} \triangleq \mathrm{PRF}^{\mathtt{I}} + \mathrm{PRF}^{\mathtt{E}}$ is the pilot reuse factor, which is clearly explained in \cite{Mohammadi:TC:2024}. Considering the vulnerability of information transmission to PC~\cite{2024:Mohammadi:survey}, we assume that $\tau>K_{\mathtt{I}}$ and then assign the first $K_{\mathtt{I}}$ out of the $\tau$ orthogonal pilot sequences to the $K_{\mathtt{I}}$ IRs as a priority (i.e., $\mathrm{PRF}^{\mathtt{I}}=0$), with the remaining $(\tau - K_{\mathtt{I}})$ pilot sequences then allocated to the ERs. We denote the SIM's PS benchmarks that adopt stochastic PS, equal PS, and heuristic PS per layer as \textbf{RDPS}, \textbf{EQPS}~\cite{Chien2020MassiveMC}, and \textbf{HPS}~\cite{hua:icc:2025}, respectively. 
We also incorporate the \textit{low-complexity} random AP mode operation with equal power allocation (\textbf{RAPEPA}) as a benchmark to evaluate the effectiveness of the proposed optimization algorithms. 
To ensure a fair comparison between the proposed \textbf{JAPPA} and \textbf{RAPEPA}, the number of I-APs and E-APs in \textbf{RAPEPA} is set to be aligned with the output of \textbf{JAPPA} in each network realization. Mathematically, let $M_{\mathtt{I}}^{\textbf{JAPPA}}$ and $M_{\mathtt{E}}^{\textbf{JAPPA}}$ denote the number of APs designated as I-APs and E-APs by the \textbf{JAPPA} scheme, respectively, where $M = M_{\mathtt{I}}^{\textbf{JAPPA}} +M_{\mathtt{E}}^{\textbf{JAPPA}}$. We denote the number of I-APs and E-APs designated by the \textbf{RAPEPA} scheme as $M_{\mathtt{I}}^{\textbf{RAPEPA}}$ and $M_{\mathtt{E}}^{\textbf{RAPEPA}}$, which stochastically sample from a discrete uniform distribution as $M_{\mathtt{I}}^{\textbf{RAPEPA}} \sim \mathcal{U}\big( M_{\mathtt{I}}^{\textbf{JAPPA}} - \delta, M_{\mathtt{I}}^{\textbf{JAPPA}} + \delta \big)$ and $M_{\mathtt{E}}^{\textbf{RAPEPA}} \triangleq M - M_{\mathtt{I}}^{\textbf{RAPEPA}}$,
where $\delta$ is a small integer threshold to allow slight variations in the assignment. The power allocation in \textbf{RAPEPA} is then set equally for all $M$.

\vspace{-0.7em}
\subsection{Impact of SIM Settings on SWIPT CF-mMIMO Performance}~\label{sec:simsettings}
In this section, we analyze the impact of SIM parameters on the minimum SE and sum-HE performance in a SIM-assisted CF-mMIMO SWIPT system. The effectiveness of EM wave manipulation by SIMs is governed by multiple factors, including the metasurface thickness $T_{\mathrm{SIM}}$, the number of layers ($L$), the number of reflective elements per layer ($N$), and the geometric design of the SIMs. Given that these parameters directly influence SWIPT performance, their optimal settings must be carefully evaluated based on our proposed CF-mMIMO SWIPT framework~\cite{an:sim:hardware1, Li:TC:2024}. 
To ensure our model remains physically consistent with the passive nature of the SIM-assisted propagation, we strictly enforce the spectral norm constraint $\Vert\mathbf{H}^{m\ell} \Vert < 1$ for every settings of the SIMs. This condition reflects the physical principle that the received signal energy cannot exceed the transmitted signal energy \cite{Nerini:CL:2024}.
We set up the layout of the PS elements per metasurface as $S_y \times S_z$ where $S_z = 4$. The spacing between adjacent antennas or elements is $\lambda/2$. 
Table~\ref{tab:thickness} presents the values of $\big\Vert \qH^{m \ell} \big\Vert \!<\! 1$ for the different SIM thicknesses considered. For instance, as $T_{\mathrm{SIM}} = 3\lambda$ and $S = 40$ make $\big\Vert \qH^{m \ell} \big\Vert \!>\! 1$ for any $l\in \mathcal{L}$, we exclude that parameter's value in the numerical simulations. 
\begin{table}[t!]
    \caption{Values of $\big\Vert \qH^{m \ell} \big\Vert$ versus various SIM thicknesses} 
    \vspace{-0.7em}
    \centering
    \begin{tabular}{|c|c|c|c|c|c|c|}
        \hline
         & $10\lambda$ & $8\lambda$ & $6\lambda$ & $5\lambda$ & $4\lambda$ & $3\lambda$ \\  
        \hline
        $S = 40$  & 0.9985   & 0.9993   & 0.9961   & 1.0351  & 1.0113 & 1.1653   \\  
        \hline
        $S = 36$  & 0.9966   & 0.9986   & 0.9979   & 0.9989  & 0.9996 & 1.1128   \\  
        \hline
        $S = 32$  & 0.9956   & 0.9982   & 0.9992   & 0.9976  & 0.9906 & 1.1222   \\  
        \hline
        $S = 28$  & 0.9954   & 0.9981   & 0.9992   & 0.9976  & 0.9882 & 1.0952   \\  
        \hline
    \end{tabular}
    \vspace{-0.1em}
    \label{tab:thickness}
\end{table}

Figure~\ref{fig:SE_v_thickness} and ~\ref{fig:HE_v_thickness} showcase the impact of SIMs thickness on the CF-mMIMO SWIPT performance. In particular, $T_{\mathrm{SIM}}$ directly affects the adjacent metasurface spacing $d^{\ell}_{s, \breve{s}} \triangleq T_{\mathrm{SIM}}/L$, which provides additional DoF for the EM wave manipulation. 
As the $d^{\ell}_{s, \breve{s}}$ expands, the propagation deteriorates due to the \textit{rank deficiency} of the propagation matrix $\FmRIS$. Under the \textbf{RDPS} scheme, the minimum SE and sum-HE degrade by $16$\% and $76$\%, respectively, as $L$ increases. In contrast, the sub-optimal \textbf{EQPS} scheme achieves approximately $16$\% and $105$\% gains in WIT and WPT performance, respectively, as $L$ increases. 
This can be explain that the randomly configured \textbf{RDPS} acts as a destructive propagating medium. With each additional layer, the signal undergoes further detrimental scattering and attenuation, destroying coherent beamforming gain. Conversely, \textbf{EQPS} benchmark shows improved performance as $L$ increases since equal phase shifts configuration preserves constructive forward signal propagation and leverages the array gain of the stacked architecture. The numerical results highlight the critical role of $T_{\mathrm{SIM}}$, but without an appropriate PS configuration, propagation attenuation occurs regardless of the chosen thickness settings.

\begin{figure*}[t]
    \centering
        \begin{minipage}[t]{0.32\textwidth}
        \centering
          \includegraphics[trim=0 0cm 0cm 0cm,clip,width=1.12\textwidth]{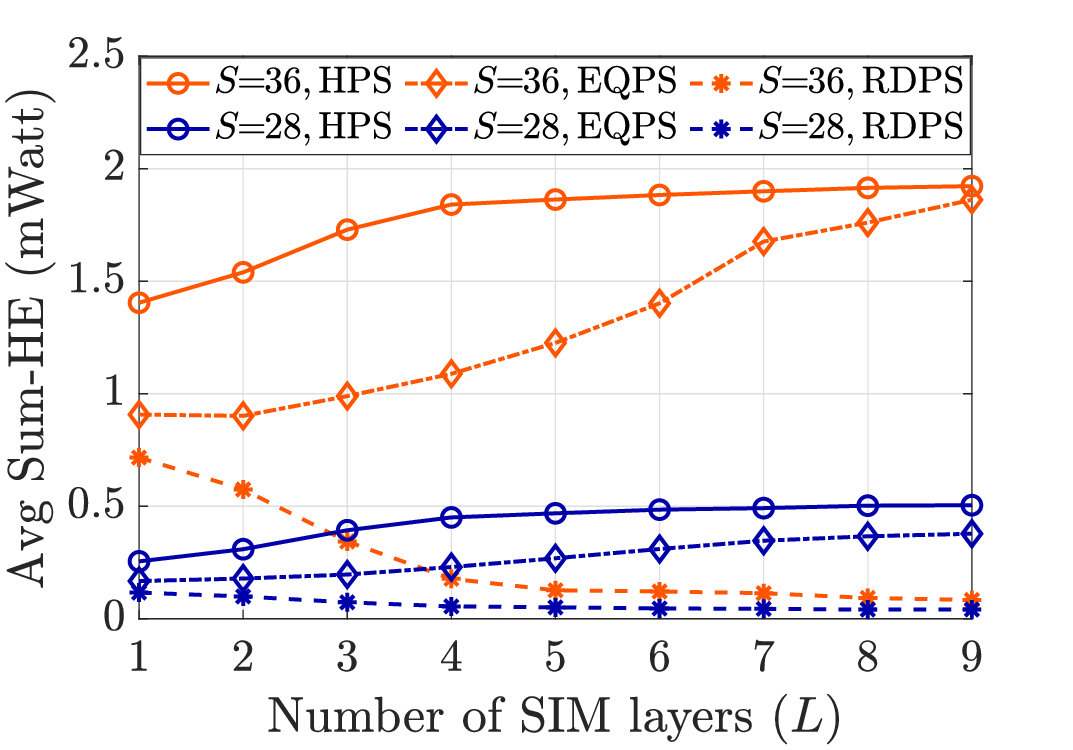}
          \vspace{-1.5em}
        \caption{\small Sum-HE versus $S$ ($T_{\mathrm{SIM}}=4\lambda$, $M = 10, \mathrm{PRF}^{\mathtt{I}} = 0, \mathrm{PRF}^{\mathtt{I}} = 3$).}
        \label{fig:figuref}
    \end{minipage}
    \hfill
    \begin{minipage}[t]{0.32\textwidth}
        \centering
          \includegraphics[trim=0 0cm 0cm 0cm,clip,width=1.12\textwidth]{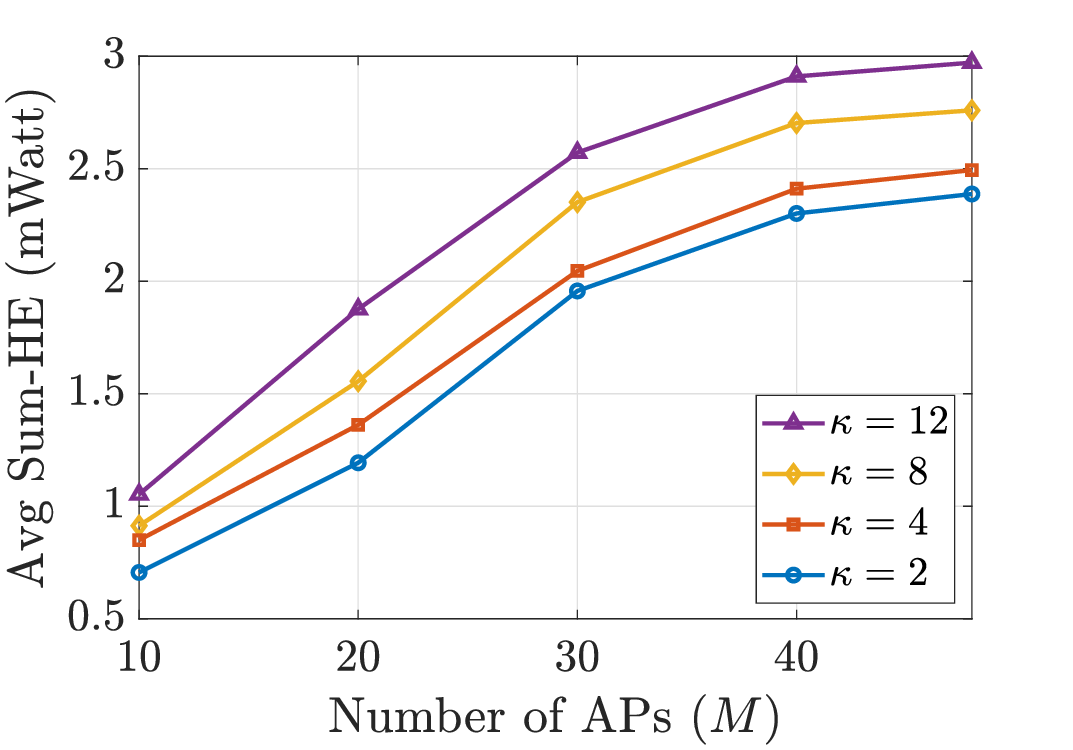}
          \vspace{-1.5em}
        \caption{\small The impact of Ricean factors on the sum-HE ($MN\!=\!480, S=36,L=2, \mathrm{PRF}^{\mathtt{I}} = 0, \mathrm{PRF}^{\mathtt{I}} = 3$).}
    \label{fig:wpt_v_ricean}
    \end{minipage}
    \hfill
    \begin{minipage}[t]{0.32\textwidth}
        \centering
          \includegraphics[trim=0 0cm 0cm 0cm,clip,width=1.12\textwidth]{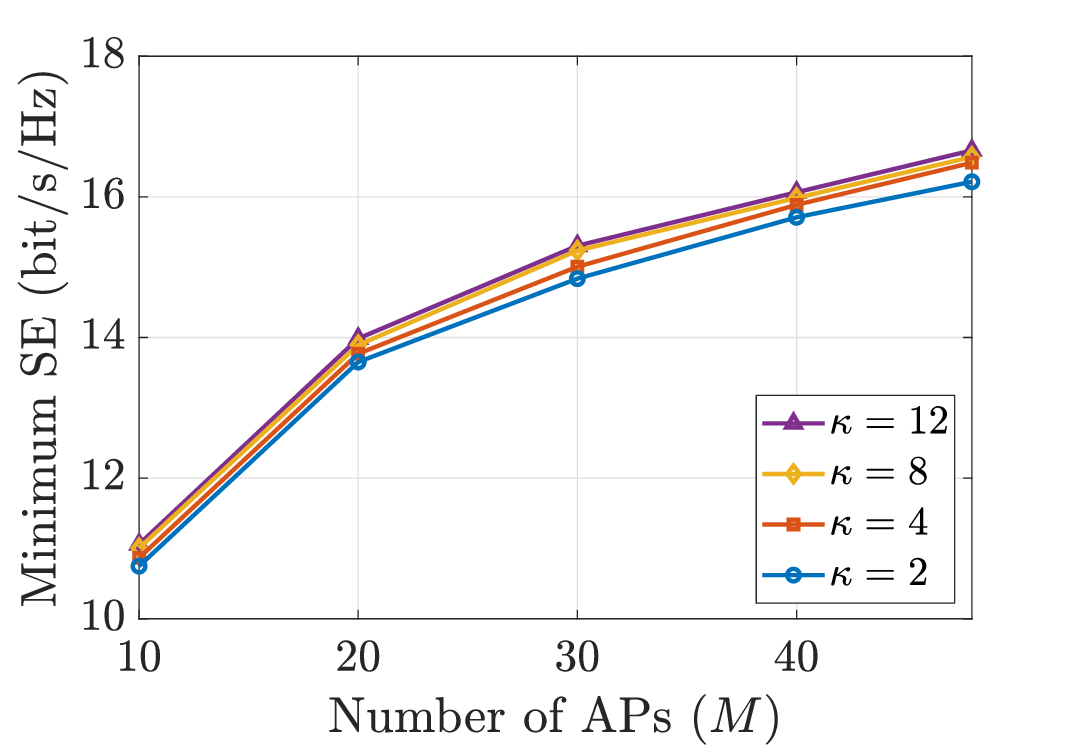}
          \vspace{-1.5em}
        \caption{\small The impact of Ricean factors on the minimum SE ($MN\!=\!480, S=36,L=2, \mathrm{PRF}^{\mathtt{I}} = 0, \mathrm{PRF}^{\mathtt{I}} = 3$).} \normalsize
        \label{fig:wit_v_ricean}
    \end{minipage}
    \hfill
\vspace{-0.7em}
\end{figure*}
\begin{figure*}[t]
\vspace{-0.1em}
    \centering
    \begin{minipage}[t]{0.32\textwidth}
        \centering
            \includegraphics[trim=0 0cm 0cm 0cm,clip,width=1.12\textwidth]{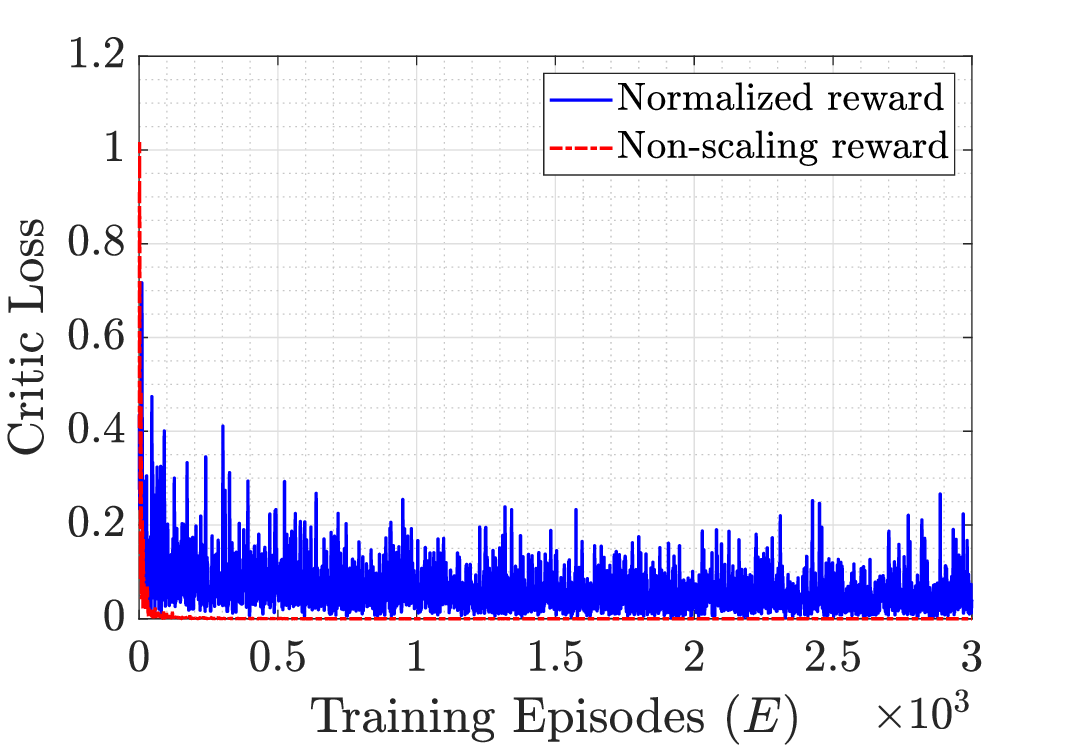}
            \vspace{-1.5em}
        \caption{\small Impact of reward structures onto the critic loss patterns ($\SEth\!=\!12$,  $M\!=\!20, \mathrm{LRC}=1e\!\!-\!\!4, \mathrm{LRA}=5e\!\!-\!\!5$).} \normalsize
        \label{fig:critic_loss}
    \end{minipage}
    \hfill
    \begin{minipage}[t]{0.32\textwidth}
        \centering
          \includegraphics[trim=0 0cm 0cm 0cm,clip,width=1.12\textwidth]{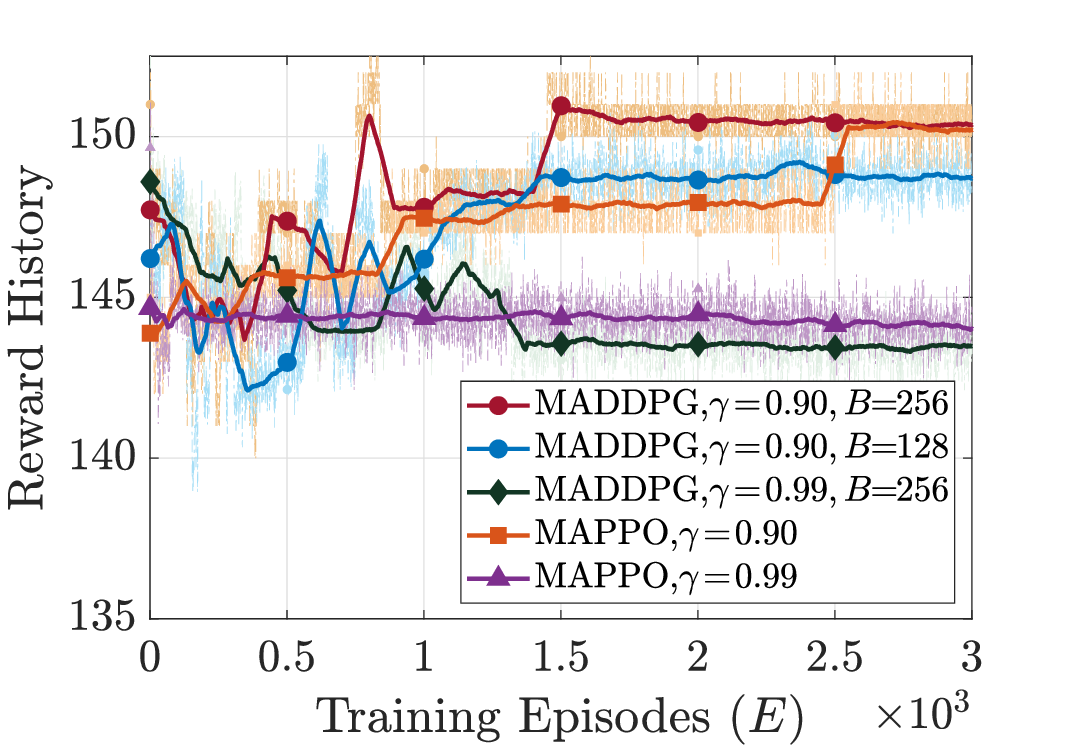}
          \vspace{-1.5em}
        \caption{\small The impact of hyperparameters on the learning pattern ($\SEth\!=\!12$,  $M\!=\!20, \mathrm{LRC}=1e\!\!-\!\!4, \mathrm{LRA}=5e\!\!-\!\!5$).}
        \label{fig:hyperparams_pattern}
    \end{minipage}
    \hfill
    \begin{minipage}[t]{0.32\textwidth}
        \centering
          \includegraphics[trim=0 0cm 0cm 0cm,clip,width=1.12\textwidth]{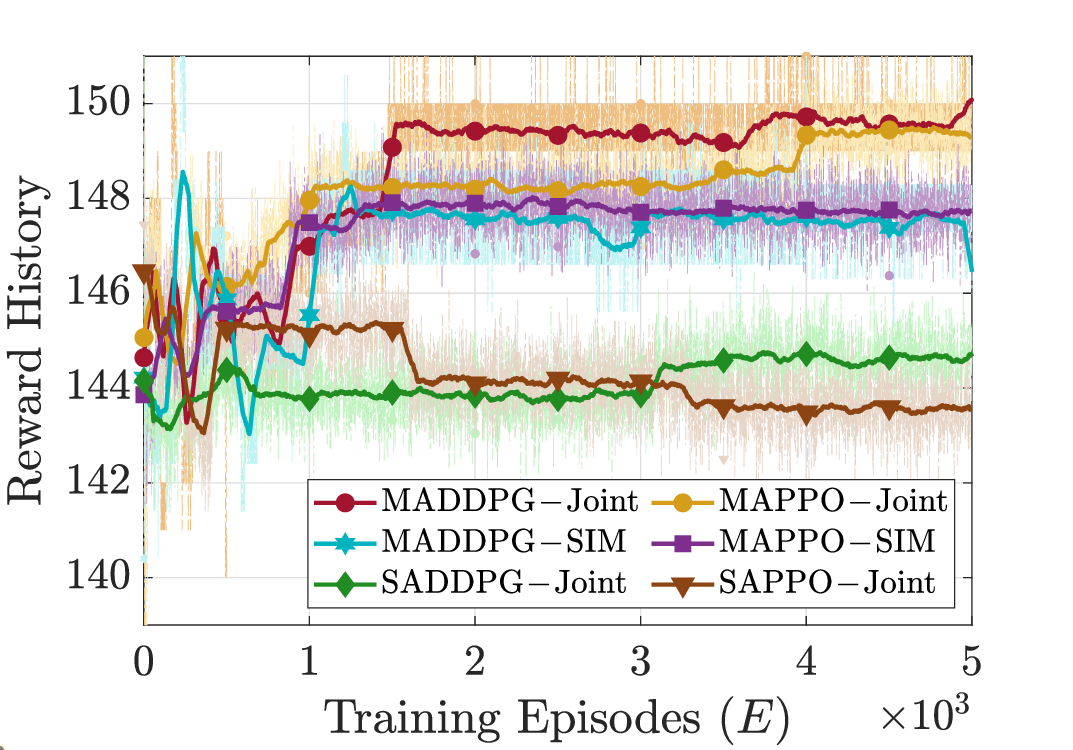}
          \vspace{-1.5em}
        \caption{\small The impact of joint local actions on the reward ($\SEth\!=\!12$,  $M\!=\!20, \mathrm{LRC}=1e\!\!-\!\!4, \mathrm{LRA}=5e\!\!-\!\!5, \bar{\gamma}=0.9, B=256$).}
        \label{fig:learning_pattern}
    \end{minipage}
\vspace{-0.5em}
\end{figure*}

Figures~\ref{fig:figurec} and~\ref{fig:figuref} show the impact of SIM PSs configuration schemes on minimum SE and sum-HE metrics. Under physical constraints, an insufficient number of elements per layer, $N$, leads to rank deficiency in the propagation channel, degrading SWIPT performance. For $L=5$ and $S=36$, the \textbf{HPS} scheme yields minimum SE and sum-HE values that are $1.05$ and $1.52$ times greater than those achieved by the \textbf{EQPS} scheme, respectively. By leveraging heuristic PS search, we can effectively reduce the number of SIM layers and PS elements while maintaining optimal DoF, thereby minimizing propagation attenuation and lowering computational and hardware requirements. It is clear that SIMs with multiple layers outperform RISs in terms of SE and HE metrics.

Figure \ref{fig:wpt_v_ricean} shows the impact of the Ricean factor on the average sum-HE in the SWIPT scenario. When $M=10$, the average sum‑HE increases from approximately 0.7 mW for $\kappa=2$ to about 1.05 mW for $\kappa=12$, corresponding to a two-fold improvement. As the number of deployed APs grows, this upward trend persists, with higher Ricean factors consistently yielding larger harvested energy and providing gains on the order of 30\% when $\kappa$ magnitude increases from 2 to 12. These results indicate that the harvested energy at the receivers is highly sensitive to the Ricean factor ($\kappa$), which aligns with the strong dependence of the rectified power on the received signal amplitude. This increase is primarily due to the more pronounced LoS component associated with higher $\kappa$ values.

Figure \ref{fig:wit_v_ricean} illustrates the impact of the Ricean factor on the achievable minimum SE. For denser AP deployment, the benefit of larger $\kappa$ becomes slightly more pronounced. At $M=48$, the minimum SE increases from around 16.2 to 16.6 bit/s/Hz when $\kappa$ grows from 2 to 12, which corresponds to a gain of almost 3\%. Thus, the minimum SE exhibits an upward trend as the Ricean factor increases, but the sensitivity is more moderate than that of the sum‑HE metric.

\vspace{-0.8em}
\subsection{Numerical Evaluation for DRL Frameworks}
We adopt the preliminaries of deep deterministic policy gradient (\textbf{DDPG})~\cite{Lillicrap:DDPG} and proximal policy optimization (\textbf{PPO})~\cite{schulman:PPO} as representative algorithms for the off-policy and on-policy learning strategies, respectively.
We determine the hyper-parameters for the training process that fit the best with our studied systems through a trial-and-error process. The critic network is made up of two hidden layers (with 1028 and 512 nodes), while the actor network includes three hidden layers (with 1028, 512, and 256 nodes). The learning rates for the actor and critic network are 1e-4 and 5e-3, respectively. We train both \textbf{DDPG} and \textbf{PPO} mechanisms for 5,000 episodes and 300 steps per episode. For the \textbf{DDPG} method, we set the batch size $\mathrm{B_e}=128$, discount factor $\bar{\gamma} = 0.99$, and buffer length $1e6$. For the \textbf{PPO} method, \textbf{PPO} clipping $0.15$ and generalized advantage estimation rate $0.995$ are applied to ensure learning stability and accurate action's value. The episodic exploration rate and the update of the target network are $1e-4$. Additionally, ReLu activation is implemented for all hidden layers and critic's output layer, while Tanh activation is applied for the actor's output layer. Eventually, gradient clipping $0.5$ is set to prevent gradient explosion and unstable training. 

Figure \ref{fig:critic_loss} illustrates the impact of the normalized reward function on the critic loss patterns, which directly reflects the stability of the long-term return during the learning progress. In this regard, we simply set the baseline step reward $r^{\rm{base}}[t]= \lambda_r\Delta_{\rm{HE}} + (1-\lambda_r)\mathsf{Q}[t] - \lambda_{\rm{SE}}$ without the max-min normalization procedure in (36). Without reward scaling, the critic’s loss showed large fluctuations during the initial episodes because the long-term target returns had high variance. Without stable backpropagation, the critic loss with non-scaling reward quickly collapses to near zero, indicating that the agent fails to learn entirely because the feedback signal becomes numerically problematic. By appropriately scaling the reward to a suitable range, the returns and advantage computations are kept within a consistent stable regime leading to convergent learning behavior.

Figure~\ref{fig:hyperparams_pattern} illustrates the impact of the batch size $B$, used in \textbf{DDPG}, and the expected long-term reward discount factor $\bar{\gamma}$, used in both \textbf{DDPG} and \textbf{PPO} frameworks. While $\bar{\gamma}$ governs the trade-off between short-term and long-term rewards, $B$ influences the stability and efficiency of gradient updates during learning. We first observe that the learning behavior in the initial phase of \textbf{DDPG} tends to be less stable compared to \textbf{PPO}. This instability can be attributed to the stochastic nature of experience replay and the Ornstein-Uhlenbeck exploration strategy, which aim to explore the solution space before transitioning to exploitation. In contrast, \textbf{PPO} exhibits more stable learning in the early stages due to their use of latest trajectory tuples, which constrains abrupt policy changes and promotes smoother convergence. 
For \textbf{MADDPG} and \textbf{MAPPO} schemes, a high discount factor of $\bar{\gamma}=0.99$ may not be ideal for episodic accumulated reward structures, even with large sampling batches. Although this setting can initially achieve the highest reward in early episodes, it tends to make the learned policies excessively prioritize distant future returns. This can result in poor gradients and cause the policies to converge towards sub-optimal solutions

\begin{table}[t]
\caption{\small Training time of the learning-based solutions (seconds per episode)}
\vspace{-1.2em}
\footnotesize
\begin{center}
\renewcommand{\arraystretch}{1}
\begin{tabular}{|r|r|r|r|r|}
\hline
\multicolumn{1}{|c|}{\multirow{2}{*}{\textbf{Benchmarks}}} & \multicolumn{4}{c|}{\textbf{Number of APs ($M$)}}                                                \\ \cline{2-5} 
\multicolumn{1}{|c|}{}                                   & \multicolumn{1}{c|}{\textbf{$10$}} & \multicolumn{1}{c|}{\textbf{$20$}} & \multicolumn{1}{c|}{\textbf{$30$}} & \multicolumn{1}{c|}{\textbf{$40$}} \\ \hline
\textbf{MAPPO}   & 35.09   & 54.88    & 61.78   & 74.26       \\ \hline
\textbf{MADDPG}  & 31.24   & 34.19    & 37.97   & 45.26       \\ \hline
\textbf{SAPPO}   & 23.58   & 24.46    & 26.19   & 28.30       \\ \hline
\textbf{SADDPG}  & 21.38   & 22.39    & 23.47   & 24.03       \\ \hline
\end{tabular}%
\vspace{-1em}
\end{center}
\label{table:execution_time}
\end{table}


\begin{table}[t]
\caption{\small Inference/execution time between learning-based and convex-based solutions (seconds)}
\vspace{-1.2em}
\footnotesize
\begin{center}
\renewcommand{\arraystretch}{1}
\begin{tabular}{|r|r|r|r|r|}
\hline
\multicolumn{1}{|c|}{\multirow{2}{*}{\textbf{Benchmarks}}} & \multicolumn{4}{c|}{\textbf{Number of APs ($M$)}}                                                \\ \cline{2-5} 
\multicolumn{1}{|c|}{}                                   & \multicolumn{1}{c|}{\textbf{$10$}} & \multicolumn{1}{c|}{\textbf{$20$}} & \multicolumn{1}{c|}{\textbf{$30$}} & \multicolumn{1}{c|}{\textbf{$40$}} \\ \hline
\textbf{JAPPA}   & 193.17   & 356.23   & 582.41   & 826.12      \\ \hline
\textbf{MAPPO}   & 0.021   & 0.026    & 0.033   & 0.036       \\ \hline
\textbf{MADDPG}  & 0.021   & 0.022    & 0.027   & 0.033       \\ \hline
\textbf{SAPPO}   & 0.073   & 0.074    & 0.077   & 0.080       \\ \hline
\textbf{SADDPG}  & 0.060   & 0.061    & 0.065   & 0.066       \\ \hline
\end{tabular}%
\vspace{-0.2em}
\end{center}
\label{table:inference_time}
\end{table}

\begin{figure}[t]
	\centering
	\vspace{-.75em}
	\includegraphics[width=0.46\textwidth]{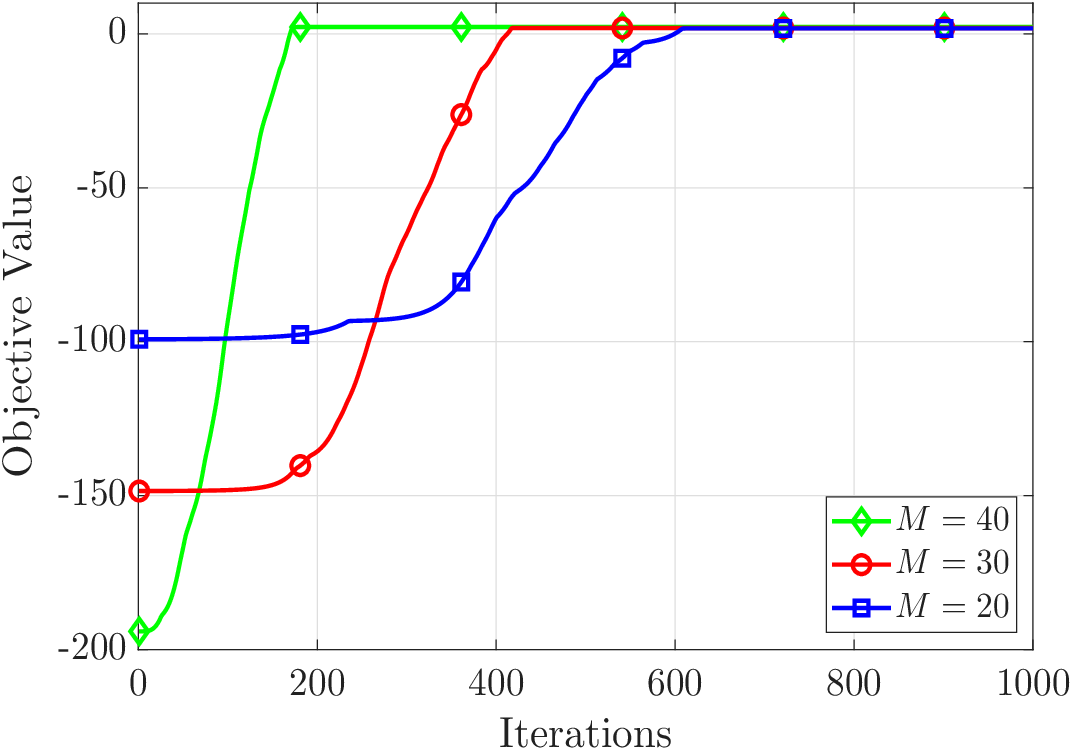}
	\vspace{-0.5em}
	\caption{The convergence pattern of the \textbf{JAPPA} solution ($N = 480/M, S=36, L=2, \SEth = 10$ [bit/s/Hz], $ K = 3, J = 4$).}
	\label{fig:SCAPattern}
\end{figure}

To this end, two scenarios for evaluating the CTDE franework are considered: (i) each AP optimizes its full local action as defined in~\eqref{eq:LocalAction}, namely \textbf{-Joint Action}, and (ii) each AP only optimizes the SIM configurations, while $\textbf{RAPEPA}$ is executed by the backhaul CPU (\textbf{-SIM}). 
Figure~\ref{fig:learning_pattern} provides the numerical results of \textbf{-Joint Action} and \textbf{-SIM} for the CTDE and CTCE strategies. For both on- and off-policy mechanisms, the \textbf{MA} framework reaches its optimal policies after 3,000 episodes, while the \textbf{SA} counterparts require 4,000 episodes to reach the phase of reward exploitation since single agent suffers from limited observability, making it challenging to explore the joint action space effectively and thus requiring more episodes to converge to stable policies. In addition, the reward increasing trend is insignificant for both CTCE mechanisms. 
This advantage can be attributed to the significantly large dimensions of $d_{go}$ and $d_{ga}$ spaces in the SA setting, which may hinder the learning efficiency due to the \textit{curse of dimensionality} and the absence of coordinated learning signals. 
As a result, \textbf{SAPPO} struggles due to limited experience reuse, high policy variance, and poor visibility in large dimensions, while \textbf{SADDPG} slightly benefits from stabilized learning through experience replay, making it more resilient to complexity.
Moreover, \textbf{-Joint Action} achieves approximately two-fold improvement in episode-accumulated reward compared to \textbf{-SIM}, averaged across CTDE settings. However, optimizing only \textbf{-SIM} yields more stable convergence and consistently satisfies sum-HE under the SE constraint. This demonstrates that optimizing SIMs' PS is crucial, as it noticeably enhances performance in CF-mMIMO SWIPT systems.


\begin{figure*}[t]
\vspace{-0.5cm}
    \centering
    \begin{minipage}[t]{0.32\textwidth}
        \centering
          \includegraphics[trim=0 0cm 0cm 0cm,clip,width=1.12\textwidth]{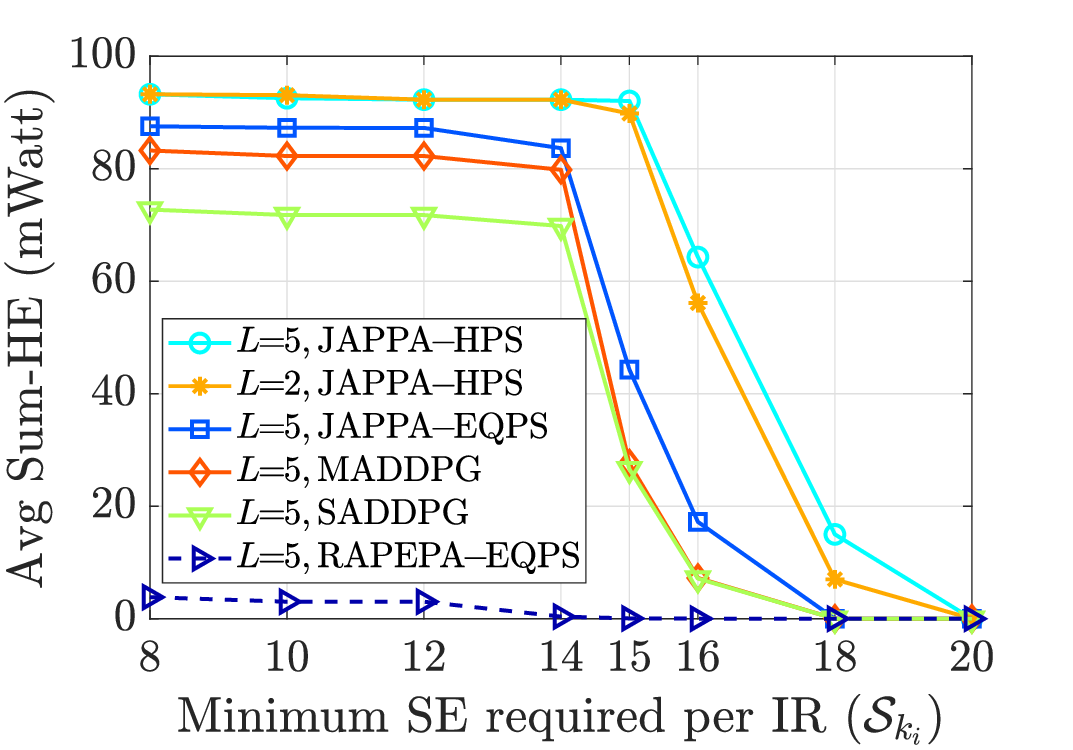}
          \vspace{-1.5em}
        \caption{\small Average sum-HE versus $\mathcal{S}_{k_i}$ ($S=36, M = 30$).\normalsize}
    \label{fig:HE_v_seth}
    \end{minipage}
    \hfill
    \begin{minipage}[t]{0.32\textwidth}
        \centering
          \includegraphics[trim=0 0cm 0cm 0cm,clip,width=1.12\textwidth]{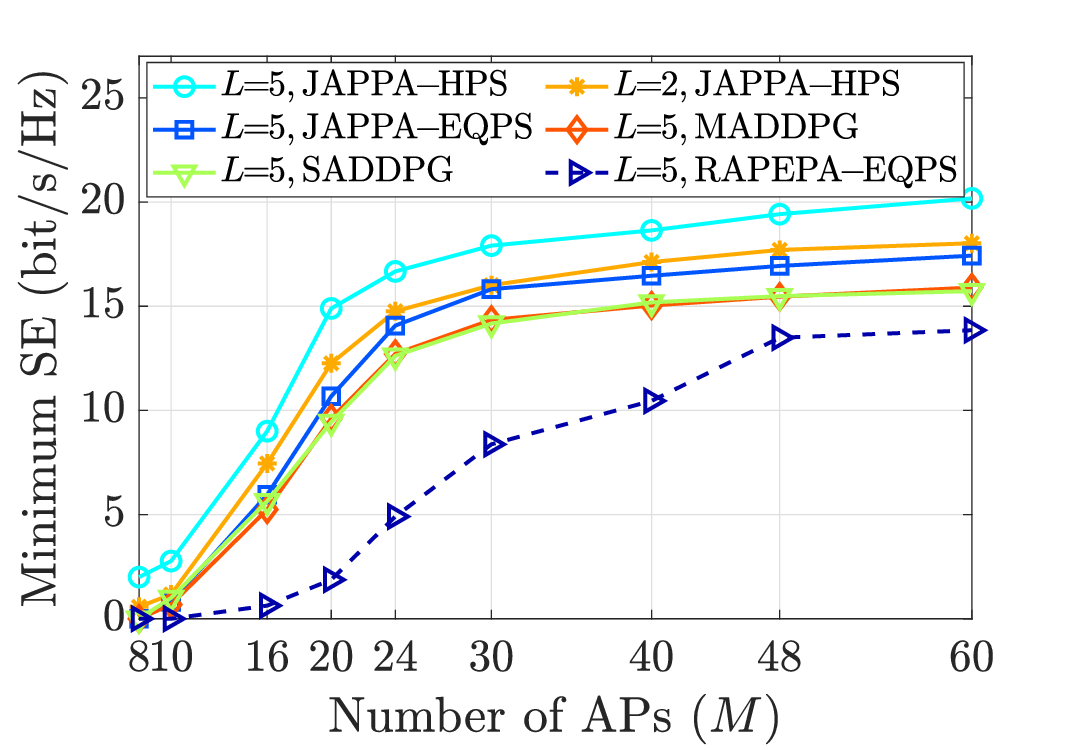}
          \vspace{-1.5em}
        \caption{\small Minimum SE versus $M$ ($\Gamma_{k_e}=0.01$ mW, $\SEth = 15$ bit/s/Hz). \normalsize}
        \label{fig:SE_v_M}
    \end{minipage}
    \hfill
    \begin{minipage}[t]{0.32\textwidth}
        \centering
            \includegraphics[trim=0 0cm 0cm 0cm,clip,width=1.12\textwidth]{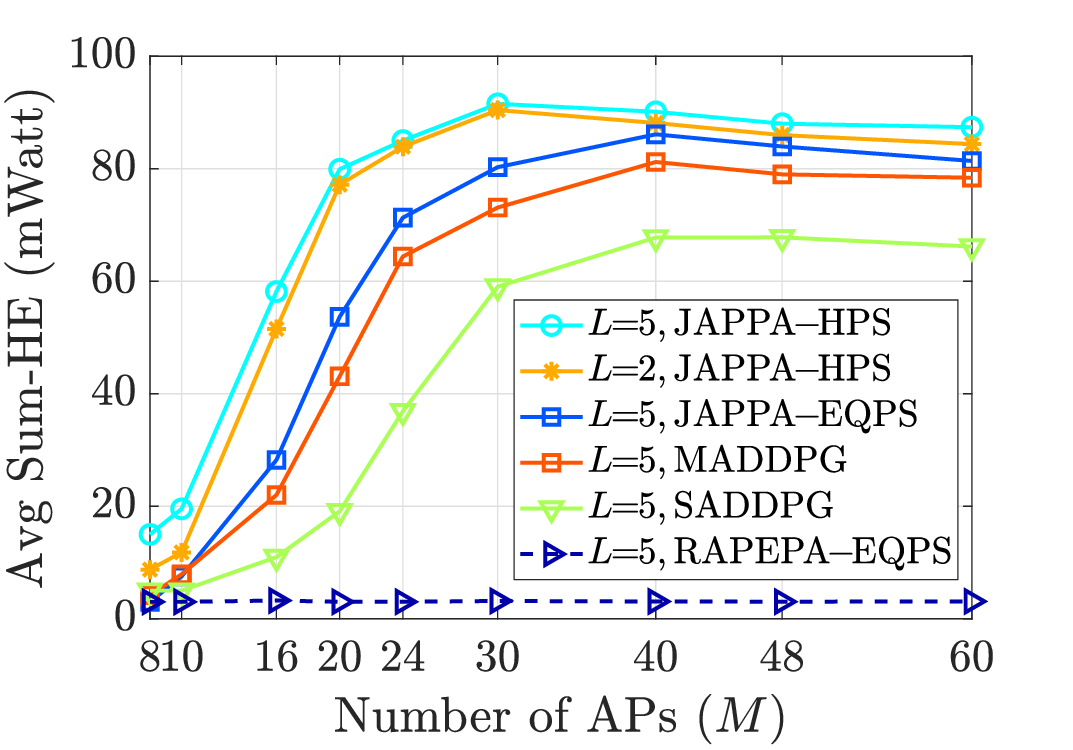}
            \vspace{-1.5em}
        \caption{\small Average sum-HE versus $M$ ($\Gamma_{k_e}=0.01$ mW, $\SEth = 15$ bit/s/Hz).\normalsize}
        \label{fig:HE_v_M}
    \end{minipage}
\vspace{-1.0em}
\end{figure*}
Table \ref{table:execution_time} compares the average training time per episode for the learning-based frameworks as the number of APs scales. 
We observe that the CTCE schemes, i.e., \textbf{SADDPG} and \textbf{SAPPO}, exhibit faster training times compared to their CTDE counterparts, \textbf{MADDPG} and \textbf{MAPPO}. This is because the CTCE framework trains a single pair of global actor-critic networks, whereas MA frameworks correspondingly update $M$ actor networks, introducing significant computational overhead from managing multiple optimizers and gradient backpropagations. Among the CTDE frameworks, \textbf{MAPPO} requires significantly longer training time than \textbf{MADDPG} due to its on-policy nature, which involves computationally intensive calculations for GAE factor and entropy regularization, while the deterministic single-gradient updates of the off-policy \textbf{MADDPG} require a simpler back-propagation process \cite{cite:thien_iotj_2023}.

Table \ref{table:inference_time} illustrates the online inference (execution) time of the trained DRL models against the convex-based \textbf{JAPPA} benchmark. The main advantages of the learning-based solutions compared to their convex-based counterparts are their ability to make instantaneous decisions after obtaining the optimal parameterized models. The latency in \textbf{JAPPA} arises from its nature to iterate over all high-dimensional SCA loops until convergence. In addition, while SA schemes exhibit faster training times, the MA schemes achieve superior inference speed due to the parallel distributed execution inherent to the CTDE strategy, where $M$ APs perform inference on their local $(L\times S+K+1)$-dimensional solutions, simultaneously enabled by smaller input and action dimensions per actor network, as described in details in Section \ref{sec:CTDE}.

\vspace{-1em}
\subsection{Performance Comparison Across  Optimization Methods}~\label{sec:vb}
In this subsection, we present numerical results comparing the DRL-based and convex-based solutions for enhancing CF-mMIMO SWIPT performance under the physical SIM settings described in Section~\ref{sec:simsettings}. After training, the learned policies are directly used to determine the optimal solutions, and the performance is averaged over multiple network realizations. We choose to numerically evaluate \textbf{MADDPG} and \textbf{SADDPG} as representatives of the CTDE and CTCE strategies, respectively, given the comparable SWIPT performance observed between the \textbf{DDPG} and \textbf{PPO} mechanisms. 

Figure~\ref{fig:SCAPattern} shows the convergence behavior of the \textbf{JAPPA} solution across scalable network setups. For a fair comparisons between convex-based and learning-based solutions, we initialize the same random seed $(100)$ in both PYTHON and MATLAB simulations so that the system parameters are identically generated before the optimization benchmarks begin. The penalty term to satisfy the binary AP mode selection constraints is initialized as $\lambda^{\rm{pen}}=10$. In general, \textbf{JAPPA} is able to provide near-optimal solutions with maximum achievable HE gain per ER for all expanding scenarios. We observe that with a larger number of deployed APs, the convex-based solutions converge faster due to an increasing number of serving E-APs to achieve the maximum sum-HE gains while the minimum SE threshold is satisfied. Conversely, a smaller number of APs ($M = 20$) decreases the computational complexity per iteration but induces longer execution time to allocate adequate number of E-APs per scenario setup. Overall, the convex-based baseline is robust, though its scalability can become a bottleneck in larger deployments.

Figure~\ref{fig:HE_v_seth} shows how different optimization schemes affect sum-HE under varying SE QoS levels. Joint optimization enables operation in the upper region of the logistic harvesting curve~\eqref{eq:NLEH:av}, whereas \textbf{RAPEPA} underperforms due to poor AP-receiver matching. As $\SEk$ increases, more I-APs are assigned, reducing sum-HE. However, integrating \textbf{JAPPA} with a 5-layer SIM yields an extra 12~mW gain even at $\SEth=18$~bit/s/Hz. In contrast, \textbf{MADDPG} and \textbf{SADDPG} struggle beyond $\SEth \geq 18$ due to their sensitivity to penalty terms $\lambda_{\mathrm{SE}}$, which destabilize reward signals and hinder training.

Figure~\ref{fig:SE_v_M} illustrates the performance of the proposed solutions versus $M$ in terms of the minimum SE metric. Under a practical QoS requirement of $\SEk = 12$~bit/s/Hz, the DRL-based \textbf{MADDPG} and \textbf{SADDPG} demonstrate strong universal approximation capability, with slightly lower min-SE values compared to the \textbf{JAPPA}. This can be attributed to the DRL agents focusing smoothly on maximizing the sum-HE, as reward stability is enhanced by not violating the SE constraint $\lambda_{\mathrm{SE}}$ too frequently during training.
Nevertheless, \textbf{MADDPG} outperforms \textbf{SADDPG} with a significant SE gain of $15.38\%$, validating the effectiveness of the proposed CTDE framework. Although \textbf{JAPPA-HPS} incurs the highest computational cost, it achieves a $125\%$ improvement over \textbf{RAPEPA-EQPS}, which represents the lowest-complexity benchmark.

Figure~\ref{fig:HE_v_M} illustrates the WPT performance across various $M$ values. A slight drop in sum-HE appears for $M\geq40$ due to fewer antennas per AP, though this effect is minor compared to~\cite{Hua:WCNC:2024, Mohammadi:GC:2023} thanks to SIM-based EM wave beamforming. Performance gains from SIMs grow with more layers: at $M=30$, the five-layer \textbf{JAPPA-EQPS} outperforms the two-layer version by 35.9\%. The off-policy CTDE also performs well, closely tracking the convex benchmark and surpassing the centralized \textbf{SADDPG} by ~30.77\%. With a target $\SEth=12$ bit/s/Hz, DRL-based methods effectively manage penalty rewards, ensuring stable exploration and episodic optimization.

\section{Conclusion}~\label{sec:conclusion}
We proposed the CF-mMIMO SWIPT systems enhanced with SIMs, supported by extensive evaluation of various DRL-based algorithms. Closed-form expressions for the ergodic SE and average HE were derived to characterize the system performance. To address the high-dimensionality of the joint optimization problem and the non-convex nature of SWIPT formulations, we introduced a Markovian environment with pre-processing approaches and a normalized joint reward function to satisfy the non-convex constraints while enhance the distributed learning mechanisms.
A two-phase convex-based optimization scheme was also developed to serve as a traceability and robust benchmark for assessing DRL effectiveness. Through numerical evaluations, we investigated both the limitations of conventional centralized, whether heuristic, SCA- or DRL-based, optimization schemes and the advantages of the proposed CTDE framework in mitigating these drawbacks. Specifically, the DRL-based CTDE methods demonstrates superior ability to globally determine optimized sum-HE without relying on complex convex approximations post-training and effectively avoid sub-optimal convergence caused by the \textit{curse of dimensionality} inherent in fully centralized DRL-based optimization.

Future work will extend the proposed CTDE-MADRL framework to near-field SWIPT, where near-field propagation leads to strongly geometry-dependent channel responses \cite{Li_RHS_NearField_2024}. This extension primarily affects the channel modeling and the Markovian environment design, while the overall CTDE-MADRL algorithm remains applicable. In contrast, for full-duplex SWIPT, additional self-interference and cross-link interference terms arise in both UL and DL. As a result, the associated optimization becomes significantly more challenging, motivating the development of interference-aware problem formulations and enhanced learning-based solution designs. Another important direction for future research is to incorporate the effects of channel aging into the proposed system model and to further validate the learning behavior and robustness of the CTDE framework under such CSI mismatch.

\appendices
\vspace{-0.7em}
\section{Proof of~\eqref{eq:alphapzf:approx}}~\label{appendix:A1}
Given $\hatWi\!\!=\!\!\Ex\big\{\! \big( \Ghmu \big)^{\dag} \Ghmu \!\big\}$, where $\Ghmu \!\in\! \C^{N \times K_{\mathtt{I}}}\triangleq\big\{ \hgmki \big\}, \forall k_i \in \Kii$, with $\hgmki \in \C^{N \times 1}$ distributed as $\hgmki\sim\CN\big(\gmkibar, \Covhatgmki \big)$, we derive $\big[ \hatWi \big]_{k_i, k_i}$ as
\vspace{-0.2em}
\begin{align}
    \!\!=\! \Ex\big\{ \!
    \big\Vert \! \big( 
    \gmkibar \!\!+\! \tildehgmki 
    \big)& \! \big\Vert^{\!2} 
    \!\big\}
    \!\!\stackrel{(a)}{=}\!
    \Ex\big\{\! \big\Vert \gmkibar \big\Vert^{2} \! \big\}
    \!\!+\!\!
    \sqrt{\tau \snrul} \barbetamk \FmRIS^{\dag} \FmRIS \qA_{m ki}
    \nonumber\\
    &\!\!\stackrel{(b)}{\approx}\!
    \kappa \barbetamki (\zmkilos)^{\dag} \FmRIS \FmRIS^{\dag} \zmkiplos + \barbetamki,
\end{align}
where $\tildehgmki \sim \CN\big(0, \Covhatgmki \big)$ is the NLoS component of $\hgmki$, (a) exploits the independence property between the LoS and NLoS components, and (b) is obtained due to the negligible noise power $\Snn$ (in Watts) in the second term of the denominator of $\qA_{m k_i}$. Then, $\qD^{\mathtt{I}}_{m} \in \R^{K_{\mathtt{I}} \times K_{\mathtt{I}}}$ in \eqref{eq:alphapzf2:approx} is defined with the $(k_i,k_i)$-th element $\big[\qD^{\mathtt{I}}_{m} \big]_{k_i,k_i} \triangleq \barbetamki$.

\section{Proof of Proposition~\ref{Theorem:SE:PPZF}}~\label{appendix:B}
We note that the derivations for $\DSki$, $\BUki$ and $\EUIki$ in~\eqref{eq:SINE:general} follow a similar procedure as in~\cite{hua:icc:2025, Mohammadi:JSAC:2023}. Due to the consideration of PPZF precoding, the computation of \eqref{eq:IU_interference} adopts different approach. To this end, we obtain
\vspace{-0.4em}
\begin{align}~\label{eq:proofIUIb}
     &\sum\nolimits_{k_{i}'\in\Kii \setminus k_{i}}\! \!\Ex\big\{ \vert \IUIki \vert^2 \!\big\} 
    \\
    &=
    \sum\nolimits_{k'_i \in \Kii \setminus k_i} \Ex \big\{ \big\vert  \sum\nolimits_{m \in \M} \sqrt{\rho_d a_m \etamkpI}   (\hgmki)^\dag \! \wimkp^{\PZF} \big\vert^{2}  \big\} 
    \nonumber\\
    &\hspace{0em}
    +\sum\nolimits_{k'_i \in \Kii \setminus k_i} \Ex \big\{  \big\vert  \sum\nolimits_{m \in \M}  \sqrt{\rho_d a_m  \etamkpI}  (\tilgmki)^\dag \wimkp^{\PZF} \! \big\vert^{\!2}  \big\}
    \nonumber\\
    &\hspace{-0.2em} \stackrel{(a)}{=}
    \sum\nolimits_{k_{i}'\in \mathcal{P}_{k} \setminus k_i} \!\!
    \Big(
    \sum\nolimits_{m \in \M} \alphaPZFmkip
    \sqrt{a_m \snrdl \etamkpI } 
    \Big)^{2}
    \nonumber\\
    &\hspace{-0.2em} +
    \sum\nolimits_{k_{i}'\in\Kii \setminus k_i } \!\!
    \sum\nolimits_{m \in \M} \!
    a_m \snrdl \etamkpI \big( \barbetamki \trace(\FmRIS \FmRIS^{\dag} )\! -\! \gamma^{\mathtt{I}}_{m k_i} \big),\nonumber
\end{align}
where in (a), the first term follows from that the ZF suppresses interference towards all the IRs except for ones sharing the same pilot, while the second term exploits~\cite[Appendix B]{hua:icc:2025}. 

Then, we have 
\vspace{-0.5em}
\begin{align}~\label{eq:proofIUI} 
    \Ex\big\{ \vert \BUki \vert^{2} \!\big\}\!&=\!\!
    \sum\nolimits_{m \in \M} \!\!\!
    a_m \snrdl \etamkI 
    \!\Big(\! 
    \Ex \Big\{ \Big\vert \! \big(\hgmki \!\!+\!\! \tilgmki \big)^{\!\dag} \!\wimk^{\PZF} \Big\vert^{2} \! \Big\} 
    \nonumber\\
    &\hspace{4.5em}\!-
    \Ex\Big\{ \big(\hgmki \!\!+\! \tilgmki \big)^{\dag} \wimk^{\PZF} \big\} \Big\vert^{2} \Big\}
    \!\Big)
    \nonumber\\
    &\hspace{-2em}=\!\! \sum\nolimits_{m \in \M} \!
    a_m \snrdl \etamkI \big( \barbetamki \trace(\FmRIS \FmRIS^{\dag} )\! -\! \alphaPZFmki \big).
\end{align}

Subsequently, exploiting the orthogonality between $\gmkiu$ and normalized $\wemj^{\PMRT}$ due to PPZF precoding, we further obtain
\vspace{-0.3em}
\begin{align}~\label{eq:proofEUI}
    \sum\nolimits_{k_e\in \Kee} \!\!\Ex \{\vert  \EUIki \vert^2\}
    &=
    \sum\nolimits_{k_{e} \in\Kee } \!\!
    \sum\nolimits_{m \in \M} \!\!
    a_m \snrdl \etamjE \nonumber\\
    &\times\big( \barbetamki \trace(\FmRIS \FmRIS^{\dag} ) -\! \gamma^{\mathtt{I}}_{m k_i} \big).
\end{align}
To this end, by substituting the final results for,~\eqref{eq:proofIUIb}, \eqref{eq:proofIUI}, and~\eqref{eq:proofEUI} into~\eqref{eq:SINE:general}, we get the desired result in~\eqref{eq:numerical:SE:PPZF}. 

\vspace{-0.1em}
\section{Proof of Proposition~\ref{Theorem:RF:MRT}}~\label{appendix:C}
The proof of Proposition~2 applies a similar methodology as in Appendix \ref{appendix:B}, with the difference lying in the analysis of the expectation terms in \eqref{eq:El_average}.
Following the algebraic manipulation steps in~\cite[Appendix C]{hua:icc:2025}, we obtain the first term in~\eqref{eq:NLEH:av} as
\vspace{-0.2em}
\begin{align}~\label{eq:Exgmjwemj}
    &\Ex\big\{\!\big\vert\big( \gejue \big)^{\dag} \wemj^{\PMRT}\big\vert^2\!\big\}
    = \Ex\big\{\! \big\vert
    \big(
    \hgmki + \tilgmki
    \big)^{\dag}
    \wemj^{\PMRT}
    \big\vert^2\!\big\}
    \nonumber\\
    &=
    \big(\big[ \hatWe \big]_{k_e k_e} \big)^{-1}
    \Ex\big\{ 
        \big\vert
        \hgmjueH \qB_{m} \hgmjue
        \big\vert^{2}
    \big\}
    \nonumber\\
    &+
    \big( \alphaPMRTmke \big)^{-1}
    \Ex\big\{\trace\big( \hgmjue \hgmjueH \qB_{m}^{\dag} \gtilmjeu \gtilmjeuH \qB_{m} \big) \big\}
    \nonumber\\
    &\stackrel{(a)}{=}
    \big( \alphaPMRTmke \big)^{-1} 
    \Big[
    \bar{a}_{m, k_e}
    \!+\!
    \Big( 
    \kappa \barbetamke \trace\big(\FmRIS \FmRIS^{\dag} \zmkelos (\zmkelos)^{\dag}  \big) 
    \nonumber\\
    &\hspace{2em}+ 
    \gameumj
    \Big)\times
    \Big(
    \barbetamke \trace\big( \FmRIS \FmRIS^{\dag} \big) - \gameumj
    \Big)
    \Big],
\end{align}
where the second term in (a) exploits the independent property between $\hgmjue$, $\gtilmjeu$. 
Under the asymptotic approximation of $\Ex\{\qB_{m}\}$, for $\forall k'_e \in \mathcal{P}_{k} \!\setminus\! k_e$, we obtain
\vspace{-0.3em}
\begin{align}~\label{eq:Exgmjwemjp1}
    \Ex\big\{\!\big\vert\big( \gejue \big)^{\dag} \wemjp^{\PMRT}\big\vert^2\!\big\}
    &= \big( \alphaPMRTmke \big)^{\!-1} 
    \nonumber\\
    & \hspace{-8em}   
    \!\Big[
    \bar{b}_{m, k_e k'_e}
    \!+\!
    \Big( 
    \kappa \barbetamkep \trace\big(\FmRIS \FmRIS^{\dag} \zmkeplos (\zmkeplos)^{\dag}  \big)+ 
    \gameumj
    \Big)
    \nonumber\\
    &\hspace{-8em}
    \times
    \Big(
    \barbetamke \trace\big( \FmRIS \FmRIS^{\dag} \big) - \gameumj
    \Big)
    \Big],  
\end{align}
while, $\forall k'_e \notin \mathcal{P}_{k} \!\setminus\! k_e$, we have
\vspace{-0.1em}
 \begin{align}~\label{eq:Exgmjwemjp2}   
   \Ex\big\{\!\big\vert\big( \gejue \big)^{\dag} \wemjp^{\PMRT}\big\vert^2\!\big\}
     &\stackrel{(b)}{=}
    \big( \alphaPMRTmke \big)^{-1} 
    \nonumber\\
    &\hspace{-8em}
    \times
    \Big[
    \Big(
    \kappa \barbetamke \trace\big(\FmRIS \FmRIS^{\dag} \zmkelos (\zmkelos)^{\dag} \big)
    + \barbetamke \trace\big(\FmRIS \FmRIS^{\dag} \big)
    \Big)
    \nonumber\\
    &\hspace{-8em}
    \times\!
    \Big(\!
    \kappa \barbetamkep \trace\big(\FmRIS \FmRIS^{\dag} \zmkeplos (\zmkeplos)^{\dag}  \!\big)
    \!+ \!
    \gameumjp
    \Big)
    \Big],
\end{align}
where (b) exploits the independence between the orthogonal channels from AP $m$ to ER-$k_e$ and ER-$k'_e$ where $k_e k_e \notin \mathcal{P}_k$.

3) Given the statistical independence, between $\gmjue$ and normalized $\wimk^{\PZF}$ for ER $k_e$ and IR $k_i$, respectively, the third term in~\eqref{eq:NLEH:av} is obtained as
\vspace{-0.1em}
\begin{align}~\label{eq:thirdExp}
    \Ex\big\{ \big(\gejue\big)^{\!\dag} &\Ex\big\{\wimk^{\PZF} \wimk^{\PZF}\big\} \gejue \!\big\}
    = \frac{1}{N}
    \Big(
    \barbetamke \trace\big( \FmRIS \FmRIS^{\dag} \big)
    \nonumber\\
    &
    +
    \kappa \barbetamke \trace\big( \FmRIS \FmRIS^{\dag} \zmkelos (\zmkelos)^{\dag}  \big)
    \Big).
\end{align}
By substituting the final results for~\eqref{eq:Exgmjwemj},~\eqref{eq:Exgmjwemjp1}, ~\eqref{eq:Exgmjwemjp2},  and~\eqref{eq:thirdExp} into~\eqref{eq:El_average}, after some manipulations, we get the result in~\eqref{eq:El_average:PMRT}. 


\end{document}